\documentclass[a4paper,12pt]{article}
\usepackage[margin=1in]{geometry}
\usepackage{titling}
\usepackage{titlesec}
\titlespacing*{\section}{0pt}{10pt}{5pt}
\titlespacing*{\subsection}{0pt}{0pt}{15pt}

\titleformat{\section}
{\normalfont\normalsize\bfseries}{\thesection}{1em}{}

\titleformat{\subsection}[runin]
{\normalfont\normalsize\bfseries\itshape}{\thesubsection}{1em}{}

\usepackage{parskip}
\usepackage{indentfirst}

\usepackage{chngcntr}
\usepackage{apptools}
\AtAppendix{\counterwithin{theorem}{section}}

\usepackage[nokeyprefix]{refstyle}
\usepackage{varioref}
\usepackage{xr-hyper}

\usepackage[noabbrev,capitalize]{cleveref}
\crefformat{section}{\S#2#1#3}
\crefformat{subsection}{\S#2#1#3}
\crefformat{subsubsection}{\S#2#1#3}

\usepackage{bm}
\newcommand{\mb}{\bm}
\newcommand{\counteralloc}{P_{\alpha, L}}
\usepackage{amsmath}
\newcommand{\numerator}{P_{\alpha, L}(\mb A_{i, -j} | A_{ij} = a, \mb L_i)}
\newcommand{\denominator}{f_{\mb A| \mb L, i}(\mb A_i | \mb L_i)}
\newcommand{\vect}[1]{\boldsymbol{\mb{#1}}}
\newcommand{\denominatorg}{f_{\mb A| \mb L, i}(\mb A_i | \mb L_i;\vect\gamma)}
\newcommand\numberthis{\addtocounter{equation}{1}\tag{\theequation}}
\newcommand{\numeratorF}{P_{F_\alpha, L}(\mb A_{i, -j} | A_{ij} = a, \mb L_i)}
\usepackage{graphicx}
\usepackage{float}
\usepackage{url}
\usepackage{color}

\title{\vspace{-1cm} \Large \bf  Causal inference for interfering units with cluster and population level treatment allocation programs \vspace{20pt}}
\author{Georgia Papadogeorgou$^*$, Fabrizia Mealli$^\dagger$, Corwin M. Zigler$^*$ \vspace{-3.5cm}}
\date{}

\allowdisplaybreaks

\usepackage{natbib}
\usepackage{amsthm}
\newtheoremstyle{style}
{5pt} 
{5pt} 
{\itshape} 
{} 
{\bfseries} 
{.} 
{.5em} 
{} 

\theoremstyle{style}
\newtheorem{assumption}{Assumption}
\newtheorem{theorem}{Theorem}
\newtheorem{proposition}{Proposition}
\newtheorem{lemma}{Lemma}

\newtheorem{example}{Example}
\numberwithin{example}{section}
\numberwithin{equation}{section}

\begin{document}

\begin{titlingpage}
\maketitle
\vspace{-40pt}
\begin{center}
	\small{
	$^*$ Department of Biostatistics, Harvard T.H. Chan School of Public Health, Boston MA, USA. \\
	$^{\dagger}$ Department of Statistics, Informatics, Applications, University of Florence, Florence, Italy.
	}
\end{center}
\vspace{20pt}

\begin{abstract}
Interference arises when an individual's potential outcome depends on the individual treatment level, but also on the treatment level of others. A common assumption in the causal inference literature in the presence of interference is partial interference, implying that the population can be partitioned in clusters of individuals whose potential outcomes only depend on the treatment of units within the same cluster. Previous literature has defined average potential outcomes under counterfactual scenarios where treatments are randomly allocated to units within a cluster.
However, within clusters there may be units that are more or less likely to receive treatment based on covariates or neighbors' treatment. We define new estimands that describe average potential outcomes for realistic counterfactual treatment allocation programs, extending existing estimands to take into consideration the units' covariates and dependence between units' treatment assignment.  We further propose entirely new estimands for population-level interventions over the collection of clusters, which correspond in the motivating setting to regulations at the federal (vs. cluster or regional) level.  We discuss these estimands, propose unbiased estimators and derive asymptotic results as the number of clusters grows.
Finally, we estimate effects in a comparative effectiveness study of power plant emission reduction technologies on ambient ozone pollution.
\end{abstract}

\end{titlingpage}

\section{Introduction}
Most causal inference literature assumes that a unit's potential outcome depends solely on its treatment, and does not depend on the treatments of other units in the population. However, this assumption is often not reasonable. Perhaps the most classical example arises in vaccination studies \citep{Ali2005,Hudgens2008} where a unit's disease status depends on their own vaccination status but also on the vaccination status of others in their social network.  The presence of interference can lead to misleading results for familiar causal estimands \citep{Sobel2006}, or estimands that lack clear causal interpretation \citep{Tchetgen2012}, but can also introduce new estimands of intrinsic scientific interest.  

\cite{Sobel2006} defined estimands for interference when the population can be partitioned into clusters for which a unit's potential outcomes depend only on the treatment of units within the same cluster. Such assumption is called \textit{partial interference}, and the interference clusters are also called interference groups. \cite{Hudgens2008} formalized causal inference in the presence of interference in the context of two-stage randomization designs, which was extended to observational studies by \cite{Tchetgen2012}. 

In order to continue development in the context of observational studies, we highlight a key distinction that arises when formulating average potential outcomes in the presence of interference, which generally requires consideration of vectors of treatment assignments.  We use the term \textit{treatment allocation strategy} to refer to a process giving rise to either observed or hypothesized vectors of treatment assignments. The {\it observed treatment allocation strategy} refers to that which gives rise to observed treatments. The {\it counterfactual treatment allocation strategy} refers to how treatments may have been assigned in some hypothesized counterfactual world for which causal contrasts can be considered. This distinction between observed and counterfactual treatment allocation programs helps illuminate that existing causal estimands, such as those in \cite{Tchetgen2012}, are limited to counterfactual treatment allocation programs that remain agnostic with regard to covariate information (as would be the case in a two-stage randomized study). These estimands ignore the possible role of unit-level covariates that relate to treatment adoption, implicitly assuming an intervention manipulating \textit{each individual unit's treatment propensity}.  Consequently, these estimands pertain to counterfactual worlds where, for example, treatments are allocated to units according to a Bernoulli distribution with equal probability for each unit within a cluster.

In many settings, however, treatment allocations corresponding to unit-level manipulation are difficult to conceive.  For example, policy interventions may be designed to increase the regional prevalence of a treatment without direct control over the individual treatment propensity.  In such settings, individual treatment adoption might generally depend on unit-level covariates or the treatment status of neighboring units.  To address such settings, we develop new causal estimands anchored to counterfactual treatment allocations that correspond to realistic regional interventions conceived at the cluster level, where a particular allocation strategy dictates the cluster-average propensity of receiving treatment \textit{without} directly specifying individual-level treatment propensities. 
Specifically, under the assumption of partial interference, we introduce estimands for counterfactual treatment allocation programs which do not assume unit-level manipulation of treatment propensities, but allow for: 1) correlation of treatment assignment within a cluster; and 2) unit-level propensities of treatment that depend on individual and group level covariates.
Note that, in focusing on new estimands for covariate-dependent counterfactual treatment allocations programs, our work has commonalities with independent work in \cite{Barkley2017a}.

In addition to extending existing estimands to accommodate realistic regional treatment allocations, a key contribution of this work provides entirely novel estimands motivated by the desire to investigate interventions which take place at the population (vs. regional or cluster) level. 
These estimands and can be particularly useful for evaluating policies that are not designed to manipulate individual or cluster-average treatment propensity, but rather change the distribution of cluster-average propensities of receiving treatment by, for example, providing a population-wide incentive to adopt treatment.  These estimands are for counterfactual treatment allocation strategies defined {\it at the population level} to shift the distribution of the cluster-average propensity of receiving the treatment, without specifying the average treatment propensity of any specific cluster.

Definition of the new causal estimands described above is accompanied here by new estimators and derivation of corresponding asymptotic properties as the number of clusters grows.  Related work can be found in \cite{Ferracci2012}.  Other relevant work includes \cite{Liu2014} where asymptotic results are derived for growing number of clusters or number of individuals within clusters, \cite{Perez2015} where large sample variance estimators for the estimator of \cite{Tchetgen2012} are derived, and \cite{Liu2016}, where estimands and estimators are extended to the case of a network where partial interference does not hold, but asymptotic results are derived under the assumption of partial interference.

The motivating context for this work is the evaluation of interventions to limit harmful pollution from power plants that are geographically clustered.  The movement of air pollution through space leads to interference: intervening on one power plant can affect the air pollution surrounding nearby power plants.  Existing estimands such as those in \cite{Tchetgen2012} represent quantities for counterfactual treatment allocations in two steps where 1) a constant treatment probability governs the proportion of power plants that would be ``treated'' within a cluster, and 2) based on that probability, power plants within the cluster are randomly and independently assigned the treatment.  However, this structure does not cohere to that of air pollution regulations, where, in reality, the adoption of treatments at power plants is not directly mandated and is heavily influenced by power-plant characteristics (e.g., the size or operating capacity of the plant).  Instead, regulatory programs often work by incentivizing regions of power plants to adopt certain technologies (e.g., by changing the penalties for over-emission), but which power plants actually adopt them is highly dependent on covariates and may be spatially correlated. 
Additionally, estimands at the population level could refer to counterfactual situations where some higher level of government (e.g., federal) issues incentives for power plants to install the technologies, but cannot mandate installation, and different regions can comply to different degrees. 
Thus, new estimands for counterfactual treatment allocations where individual-level treatment adoption depends on covariates for population-level incentives
cohere more closely to the realities of air pollution regulations. 
The new estimators are deployed here to an analysis of U.S. power plants investigating the comparative effectiveness of Selective Catalytic or non-Catalyitic Reduction systems (relative to other strategies) for reducing ambient ozone pollution. A preliminary investigation of these same data in \cite{Papadogeorgou2018} ignored interference and indicated that these systems causally reduced NO$_x$ emissions (an important precursor to ozone pollution) but did not lead to a reduction in ambient ozone.  The analysis here to address the possibility of interference produces meaningfully different results that are more consistent with the literature relating NO$_x$ emissions to ambient ozone pollution. Note that, despite the focus on air pollution interventions, similar considerations could be construed in more classical interference settings such as vaccine studies, where certain types of community members might be more likely to receive the vaccine and vaccine programs may be designed to increase vaccine coverage at the community, or national level.

In \cref{sec:notation_estimand} we introduce the notation and the new estimands for the cluster-level intervention, followed by the definition of estimands for the population-level intervention in \cref{sec:est_aver_over_alpha}. Estimators are presented in \cref{sec:estimator}, along with unbiasedness, consistency and asymptotic normality results for an increasing number of clusters.
The rest of the paper presents some simulations in \cref{sec:simulations}, our data application in \cref{sec:data_application} and concludes with some discussion on the limitations and future directions of this paper in \cref{sec:discussion}.


\section{Estimands under partial interference}
\label{sec:notation_estimand}

We adopt the notation used in \cite{Tchetgen2012}. Let $N$ be the number of clusters, and $n_i$ the number of units in cluster $i$, $i \in \{1, 2, \dots, N\}$. Furthermore, denote $\mb A_i = (A_{i1}, A_{i2}, \dots, A_{in_i}) \in \mathcal{A}(n_i)$ to be the cluster treatment vector, and $\mb A_{i, -j} = (A_{i1}, A_{i2}, \dots, A_{ij-1}, A_{ij+1}, \dots, A_{in_i}) \in \mathcal{A}(n_i - 1)$ to be the treatment of all units in cluster $i$ apart from unit $j$, where $\mathcal{A}(n) = \{0, 1\}^{n}$.
Furthermore, let $L_{ij}$ be a vector of individual and cluster-level covariates, and $\mb L_i = (L_{i1}, L_{i2}, \dots, L_{in_i})$ be the collection of covariates of all units within a cluster.

Under the assumption of partial interference, the potential outcome of unit $j$ in cluster $i$ may depend on the treatment of units in cluster $i$, but not on the treatment of units in different clusters. For every $i$ we postulate the existence of group $i$'s potential outcomes $\mb{Y}_i(\cdot) = \{\mb{Y}_i(\mb a_i), \mb a_i \in \mathcal{A}(n_i) \}$, where $\mb{Y}_i(\mb a_i) = (Y_{i1}(\mb a_i), Y_{i2}(\mb a_i), \dots, Y_{in_i}(\mb a_i))$.


\subsection{Average potential outcome}

Under the assumption of partial interference, we define the individual average potential outcome for a counterfactual treatment allocation strategy with two features: 1) treatment assignment for units within a cluster is unlikely to be independent, and 2) individual covariates can be predictive of a unit's treatment probability. Let $\counteralloc$ represent the (arbitrarily specified) counterfactual treatment allocation program, specified intentionally to depend on covariates and/or allow correlated assignments within clusters. $\counteralloc$ is governed by parameters $\alpha$, which represent features of the counterfactual treatment allocation program of interest. For the purpose of this paper, we consider $\alpha$ to represent the cluster-average propensity of treatment, but alternatives are briefly discussed in \cref{sec:discussion}.

The individual average potential outcome is defined as:
\begin{align*}
\overline{Y}_{ij}^L(a ; \alpha)
= \sum_{\mb s \in \mathcal{A}(n_i - 1)}
Y_{ij}&(A_{ij} = a, \mb A_{i,-j} = \mb s)
\counteralloc(\mb A_{i,-j} = \mb s | A_{ij} = a, \mb L_i),
\numberthis
\label{eq:ind_aver_po}
\end{align*}
and represents the expected outcome for unit $j$ in cluster $i$ in the counterfactual world where treatment is assigned with respect to $\counteralloc$, but the treatment of unit $j$ is fixed to $a$. This estimand is well-defined for any fixed choice of $\counteralloc$. Based on the individual average potential outcome, group and population average potential outcomes are defined as
\begin{equation}
\overline{Y}_i^L(a; \alpha) =
\frac 1{n_i} \sum_{j = 1}^{n_i} \overline{Y}_{ij}^L(a; \alpha),
\label{eq:group_aver_po}
\end{equation}
and
\begin{equation}
\overline{Y}^L(a; \alpha) = \frac 1{N} \sum_{i = 1}^N \overline{Y}_i^L(a; \alpha) \label{eq:pop_aver_po}
\end{equation}
accordingly.

\subsection{The counterfactual treatment allocation in existing literature}

As mentioned previously, $\counteralloc$ can be arbitrarily chosen and represents the process with which treatment is assigned in the counterfactual world, driving the interpretation of all estimands. The above development has left unspecified the term $\counteralloc$ in (\ref{eq:ind_aver_po}) providing relative weights to different cluster treatment vectors in the individual average potential outcomes. The estimands in
\cite{Tchetgen2012} and \cite{Perez2015} correspond to counterfactual treatment strategies
\(\displaystyle
\counteralloc(\mb a_i | \mb L_i) = 
\prod_{j = 1}^{n_i}\alpha^{a_{ij}} (1 - \alpha)^{1 - a_{ij}}\), giving equal probability to all cluster-treatment vectors with the same number of treated units, irrespective of which those units are.
For this choice of $\counteralloc$ the estimands represent quantities in counterfactual worlds where individual treatment probability can be manipulated and units are assigned to treatment independently and with equal probability $\alpha$.

\subsection{Realistic counterfactual treatment allocation program}
\label{sec:choose_numerator}

However, in some situations, counterfactual treatment allocations can only be realistically conceived if allowed to depend on covariates or if they incorporate correlation between treatment of units in the same cluster.
In the study of power plant interventions on ambient air quality, the decision of whether to ``treat'' a power plant is at the discretion of the power company and heavily influenced by power plant covariates. Therefore, a hypothesized counterfactual treatment allocation is realistic only when such covariates are incorporated.

As an example, consider the power-plant level covariate `heat input', a proxy for the size of the power plant, and let $L_{ij}$ be the heat input of power plant $j$ in cluster $i$. Then, one specification of a counterfactual treatment allocation strategy that would acknowledge that different-sized power plants are more or less likely to adopt treatment is:
\begin{equation}
\mathrm{logit} \counteralloc(A_{ij} = 1 | L_{ij}) = \xi_i^\alpha + \delta_L L_{ij},
\label{eq:P_alpha_L}
\end{equation}
for some \textit{fixed}, pre-specified value of $\delta_L$, and $\xi_i^\alpha$ such that 
$$
\frac1{n_i} \sum_{j = 1}^{n_i} \mathrm{expit} \left( \xi_i^{\alpha} + \delta_L L_{ij} \right) = \alpha.
$$
The value $\delta_L$ here could be specified according to knowledge of how the size of the power plant is expected to impact the propensity to adopt treatment.

Based on (\ref{eq:P_alpha_L}), the probability of the cluster treatment vector under the counterfactual treatment allocation $\counteralloc(\mb A_i = \mb a_i | \mb L_i)$ could be fully specified  by hypothesizing that the $A_{ij}$'s are conditionally independent given $\mb L_i$, and $A_{ij}$ is conditionally independent of $\mb L_{i, -j}$ given $L_{ij}$. Then,
$$
\counteralloc(\mb A_i = \mb a_i | \mb L_i) =
\prod_{j = 1}^{n_i} \counteralloc(A_{ij} = a_{ij} | \mb L_i) =
\prod_{j = 1}^{n_i} \counteralloc(A_{ij} = a_{ij} | L_{ij})
$$
which, in turn, specifies $\counteralloc(\mb A_{i,-j} = \mb s | A_{ij} = a, \mb L_i)$ for all $\mb s \in \mathcal{A}(n_i - 1)$ giving relative weights in the specification of the individual average potential outcome (\ref{eq:ind_aver_po}).
Based on this specification of $\counteralloc$, the estimands of interest correspond to quantities in a hypothesized world where treatment is assigned independently across units with treatment propensity that depends on $L_{ij}$, but is on average equal to $\alpha$.

Alternatively, a counterfactual treatment allocation strategy can also be defined to incorporate dependence of treatments in the same cluster. For example, consider
$$
\mathrm{logit} \counteralloc (A_{ij} = 1 | L_{ij}, \theta_{ij}) =
\xi_i^\alpha + \delta_L L_{ij} + \theta_{ij},
$$
where $\theta_{ij}$ is a mean 0 spatial random effect with fixed correlation matrix decaying with distance. This choice of $\counteralloc$ corresponds to a counterfactual treatment allocation program that depends on covariates and incorporates dependent treatment assignment of units within a cluster. A data-driven way to choose $\counteralloc$ is presented in \cref{sec:data_application}.

\subsection{Direct and indirect effects}
\label{subsec:direct_indirect}

Different contrasts of average potential outcomes can be considered to characterize how treatment affects the outcome of interest. For counterfactual allocation strategy $\counteralloc$, direct effects represent contrasts in average potential outcomes when only the individual treatment changes. On the other hand, indirect effects contrast average potential outcomes for a fixed level of individual treatment, but different specification of the parameter $\alpha$ governing the counterfactual allocation program. For that reason, indirect effects represent expected changes in potential outcomes for changes only in the ``treatment of neighbors'', and they can be thought of as a measure of interference. Indirect effects are also known in the literature as spillover effects.

Based on the individual, group and population average potential outcomes, one can define the individual, group and population direct effects as
\begin{align*}
DE_{ij}^L(\alpha) = &\overline{Y}^L_{ij}(1;\alpha) - \overline{Y}^L_{ij}(0;\alpha), \\
DE_i^L(\alpha) = & \overline{Y}_i^L(1,\alpha) - \overline{Y}_i^L(0;\alpha)
= \frac{1}{n_i} \sum_{j = 1}^{n_i} DE_{ij}^L(\alpha) \\
DE^L(\alpha) = & \overline{Y}^L(1,\alpha) - \overline{Y}^L(0;\alpha)
= \frac{1}{N} \sum_{i = 1}^N DE_i^L(\alpha)
\end{align*}
accordingly. Similarly, the individual indirect effect is defined as
$$ IE_{ij}^L(\alpha_1, \alpha_2) = \overline Y_{ij}^L(0,\alpha_2) - \overline Y_{ij}^L(0, \alpha_1),$$
based on which group and population indirect effects can be defined. Indirect effects could be alternatively defined for individual treatment assignment $a = 1$, but here our focus is on the effect of neighbors' treatment in the areas surrounding untreated power plants.
Contrasts other than the difference can also be considered. Based on these estimands, total effects can be defined as the sum of direct and indirect effects \citep{Hudgens2008}, while similar development can lead to the definition of overall effects.


\section{Population-Level counterfactual distribution of cluster-average treatment propensity}
\label{sec:est_aver_over_alpha}

In \cref{sec:notation_estimand} we defined the individual average potential outcome for unit $j$ in cluster $i$ (and other estimands based on it) when the cluster-average propensity of treatment $\alpha$ is fixed to a counterfactual value. Those estimands correspond to quantities of interest in counterfactual worlds were one intervenes at the level of the cluster, but units within the cluster are still allowed to choose their own treatment. In this section, new individual average potential outcomes are defined, when the unit's treatment is set to $a$, but the cluster average propensity of treatment is not fixed to a specific value $\alpha$ but arises from a hypothesized distribution.

These estimands play an important role for policy interventions that occur at a high (vs. local) administrative level. For example, consider an observed distribution of cluster-average treatment propensity $\widehat{F}_\alpha$, and an intervention that takes place over all clusters incentivizing the increase of cluster treatment coverage. This intervention does not enforce a specific average propensity of treatment for each cluster separately, but leads to an overall shift in the distribution of cluster average propensity of treatment.

Let $F_\alpha(\cdot)$ denote the observed or a hypothesized distribution of cluster-average propensity of treatment. Then, define the $F_\alpha$-individual average potential outcome as
\begin{align*}
\overline{Y}_{ij}^L(a; F_\alpha) = &
\int \overline{Y}_{ij}^L(a ; \alpha) \ \mathrm{d}F_\alpha(\alpha)
\numberthis
\label{eq:ind_aver_po_over_alpha} \\
= & \sum_{\mb s \in \mathcal{A}(n_i - 1)}
Y_{ij}(A_{ij} = a, \mb A_{i,-j} = \mb s)
\int \counteralloc(\mb A_{i,-j} = \mb s | A_{ij} = a, \mb L_i) \ \mathrm{d}F_\alpha(\alpha).
\end{align*}
Thus, $\overline{Y}_{ij}^L(a; F_\alpha)$ describes the average potential outcome of unit $j$ in cluster $i$, for cluster average probability of treatment arising from $F_\alpha$.
Consequently, the $F_\alpha$-group and population average potential outcomes are defined as
\begin{align*}
\overline{Y}_i^L(a; F_\alpha) = & \frac 1{n_i}
\sum_{j = 1}^{n_i} \overline Y_{ij}^L(a; F_\alpha) \\
\overline{Y}^L(a; F_\alpha)
= & \frac1N \sum_{i = 1}^{N} \overline Y_i^L(a; F_\alpha)
\numberthis
\label{eq:Falpha_estimand}
\end{align*}
accordingly.
Although the above estimands are well-defined for a distribution $F_\alpha$ different than the observed one, $F_\alpha$ needs to have overlapping support with the empirical distribution $\widehat{F}_\alpha$ in order to reliably estimate such quantities.

Even though direct effect estimands based on the $F_\alpha$-population average potential outcome can easily be defined as $DE(F_\alpha) = \overline{Y}^L(1;F_{\alpha}) - \overline{Y}^L(0;F_\alpha)$, the contrast of $F_\alpha$-population average potential outcomes is more interesting for the indirect effect. For two hypothesized distributions of cluster-average propensity of treatment $F_\alpha^1, F_\alpha^2$, define
\begin{equation}
IE\left(F_\alpha^1, F_\alpha^2\right) = \overline{Y}^L\left(0;F_\alpha^2\right) -
\overline{Y}^L\left(0;F_{\alpha}^1\right).
\label{eq:Falpha_indirect}
\end{equation}
Then, $IE\left(F_\alpha^1, F_\alpha^2\right)$ represents the expected outcome change for control units when the distribution of cluster-average propensity of treatment changes from $F_\alpha^1$ to $F_\alpha^2$.

\section{Estimating the population average potential outcome}
\label{sec:estimator}

For a fixed choice of $\counteralloc$, we provide estimators of the population average potential outcome in (\ref{eq:pop_aver_po}), unbiasedness and consistency results, and derive the estimator's asymptotic distribution when the number of clusters increases to infinity, for a known or correctly specified parametric cluster-propensity score model (defined below). Based on these, estimators and asymptotic distributions for the superpopulation counterparts of the estimands in \cref{subsec:direct_indirect} can be acquired as demonstrated in Example \ref{app_ex:delta_method} of the supplementary materials. Proofs are in \cref{supp_sec:proofs}. Based on similar arguments, we acquire asymptotic results for the population level estimands in (\ref{eq:Falpha_estimand}) and (\ref{eq:Falpha_indirect}).

We start by making the sample cluster-level \textit{positivity}, and \textit{ignorability} assumptions:

\begin{assumption}
	Positivity.
	For $i \in \{1, 2, \dots, N\}$, the probability of observing cluster treatment vector $\mb a_i$ given cluster covariates $\mb L_i$ is denoted by $f_{\mb A| \mb L, i}(\mb A_i = \mb a_i | \mb L_i)$ and is positive for all $\mb a_i \in \mathcal{A}(n_i)$. $f_{\mb A| \mb L, i}$ is the cluster-propensity score.
	\label{ass:group_positivity}
\end{assumption}
%
\begin{assumption}
	Ignorabililty.
	For $i \in \{1, 2, \dots, N\}$, the observed cluster treatment $\mb A_i$ is conditionally independent of the set of cluster potential outcomes $\mb{Y}_i(\cdot)$ given the covariates $\mb L_i$, denoted as $\mb A_i \amalg \mb{Y}_i(\cdot) | \mb L_i$.
	\label{ass:group_ignorability}
\end{assumption}

\subsection{Estimators of the group and population average potential outcome}

\label{sec:estimator_form}
Let
\begin{equation}
	\widehat{Y}_i^L(a ; \alpha) = \frac 1{n_i} \sum_{j = 1}^{n_i}
	\frac{\numerator}{\denominator} I(A_{ij} = a) Y_{ij}
    \label{eq:est_group_average_cond}
\end{equation}
and
\begin{equation}
    \widehat{Y}^L(a; \alpha) = \frac1N \sum_{i = 1}^N
    \widehat{Y}_i^L(a; \alpha)
    \label{eq:est_pop_aver_po}
\end{equation}
where
$\denominator$ is the cluster-level propensity score for the observed treatment, and $\numerator$ is the probability of the observed treatment on units other than $j$ given $A_{ij} = a$, under the specified counterfactual treatment allocation program.

Assuming that the group level propensity score $f_{\mb A|\mb L, i}(\cdot | \mb L_i)$ is known and Assumptions \ref{ass:group_positivity} and \ref{ass:group_ignorability} hold, then $\widehat{Y}^L_i(a ; \alpha)$, $\widehat{Y}^L(a;\alpha)$ are unbiased for $\overline{Y}_i^L(a, \alpha), \overline{Y}^L(a, \alpha)$ accordingly, as defined in (\ref{eq:group_aver_po}), (\ref{eq:pop_aver_po}). Unbiasedness is derived for a fixed set of clusters with respect to the distribution of the observed treatment assignment.

The population average potential outcome (\ref{eq:pop_aver_po}) is defined as the average of the group average potential outcomes. Alternative definitions could weigh each cluster by cluster sample size (which is what the population average potential outcome of \cite{Liu2016} simplifies to under the assumption of partial interference). In \cref{supp_sec:population_estimand}, we discuss this distinction and provide an argument why an equal-weight estimand and the corresponding estimator (\ref{eq:est_pop_aver_po}) is preferable.

\subsection{Asymptotic results for $\widehat{Y}^L(a;\alpha)$ for known propensity score}

We derive the asymptotic properties of the estimator in (\ref{eq:est_pop_aver_po}) for an increasing number of clusters $N$, denoted by $\widehat{Y}^L_N(a; \alpha)$. 
Let $\widehat Y^L_N(\alpha) = \Big(\widehat{Y}^L_N(0; \alpha),\ \widehat{Y}^L_N(1; \alpha)\Big)^T$.

Assume that the $N$ clusters are a sample of an infinite superpopulation of clusters from which they are sampled randomly. Therefore
$(\mb{Y}_i(\cdot), \mb A_i, \mb L_i)$
are now independent and identically distributed random vectors, whose distribution is denoted as $F_0$. (For notational simplicity, $n_i$ is included in $\mb L_i$.)
Assuming a superpopulation of clusters, the estimands of interest no longer pertain to the sample, but must represent quantities in the population of clusters from which the sample arose. The super-population counterpart of the population average potential outcome defined in (\ref{eq:pop_aver_po}) is
\( \displaystyle
\mu_0(a, \alpha) = E_{F_0}\Big[\overline{Y}_i^L(a; \alpha)\Big],
\)
where $\overline Y_i^L(a;\alpha)$ is defined as in (\ref{eq:group_aver_po}). Super-population direct and indirect effects correspond to contrasts in $\mu_0(a, \alpha)$.

Similarly, the sample positivity and ignorability assumptions are translated to their super-population counterparts.
\begin{assumption} Super-population positivity.
There exists $\rho > 0$ such that $\denominator > \rho$ with probability 1.
\label{ass:super_positivity}
\end{assumption}
\begin{assumption}
\textit{Super-population ignorability.}
For $F_0$, $\mb A_i \amalg \mb{Y}_i(\cdot) | \mb L_i$. 
\label{ass:super_ignorability}
\end{assumption}

\begin{theorem}
Let $\boldsymbol \mu_0(\alpha) = (\mu_0(0, \alpha), \mu_0(1, \alpha))^T$.
Under Assumptions \ref{ass:super_positivity}, \ref{ass:super_ignorability}, for known propensity score, and bounded outcome (there exists $M > 0: |Y_{ij}| < M $ with probability 1),
$\widehat{Y}_N^L(\alpha)$ is consistent for $\boldsymbol \mu_0(\alpha)$ and asymptotically normal with limiting distribution
\( \displaystyle
\sqrt{N} \left(\widehat Y^L_N(\alpha) - \boldsymbol\mu_0(\alpha) \right) \overset{d}{\rightarrow} N(0, V(\boldsymbol\mu_0(\alpha))), 
\)
where
\begin{align*}
	& V(\boldsymbol\mu_0(\alpha)) = E_{F_0} \left[
    \psi(\mb y_i, \mb l_i, \mb a_i;\boldsymbol\mu_0(\alpha))
    \psi(\mb y_i, \mb l_i, \mb a_i;\boldsymbol\mu_0(\alpha))^T \right],\\
    & \psi(\mb y_i, \mb l_i, \mb a_i;\boldsymbol\mu_0(\alpha)) = 
    \left(
    \psi_{0, \alpha}\left(\mb y_i, \mb l_i, \mb a_i;\mu_0(0, \alpha)\right), \ 
    \psi_{1, \alpha}\left(\mb y_i, \mb l_i, \mb a_i;\mu_0(1, \alpha)\right) \right)^T \\
    & \psi_{a, \alpha}(\mb y_i, \mb l_i, \mb a_i;\mu_0(a, \alpha)) = 
    \frac1{n_i} \sum_{j = 1}^{n_i} \frac{\numerator}{\denominator}
    I(A_{ij} = a)Y_{ij} - \mu_0(a, \alpha).
    \end{align*}
    \label{theorem:asymptotic_normality}
\end{theorem}

The above theorem leads to the approximation
\( \displaystyle \widehat Y^L_N(\alpha) \sim
MVN_2 \left( \vect\mu_0(\alpha), N^{-1} V(\vect\mu_0(\alpha)) \right)
\) for large number of clusters.
Even if assumptions about $F_0$ are made, the elements of \( \displaystyle
V(\vect\mu_0(\alpha)) = Cov_{F_0} \left[ \left( \overline Y_i(0,\alpha), \overline Y_i(1,\alpha) \right)^T \right]
\) (see \cref{supp_sec:Vmu}) are often hard to calculate analytically. Instead, the asymptotic variance of $\widehat{Y}_N^L(\alpha)$ can be estimated using the empirical expectation
\[
\widehat{V}\left(\vect \mu \right) =
\frac 1N \sum_{i = 1}^N \left[\psi(\mb Y_i, \mb L_i, \mb A_i; \vect \mu)
\psi(\mb Y_i, \mb L_i, \mb A_i; \vect \mu)^T\right], \]
evaluated at $\vect \mu = \widehat Y_N^L(\alpha)$.
Under regularity conditions, discussed in \cite{Iverson1989}, $ \widehat{V}\left( \widehat Y_N^L(\alpha) \right)$ will be consistent for $V(\vect\mu_0(\alpha))$.
Using Theorem \ref{theorem:asymptotic_normality} one can acquire the asymptotic distribution of a contrast between $\widehat{Y}^L(0;\alpha)$, $\widehat{Y}^L(1;\alpha)$ specifying a direct effect, by an application of the multivariate delta method.

\subsection{Asymptotic results for $\widehat{Y}^L(a;\alpha)$ for estimated propensity score from a correctly-specified parametric model}

However, most of the times the propensity score is not known, and has to be estimated using the observed data.
In the next theorem, we provide the asymptotic distribution of $\widehat Y_N^L(\alpha)$ when the propensity score is estimated using a correctly specified parametric propensity score model. In this case, the cluster-propensity score for the observed treatment vector will be denoted by $\denominatorg$ where $\vect\gamma$ are the model parameters.

\begin{theorem}
\label{theorem:asymptotic_normality_estimated_ps}
Assume that assumptions \ref{ass:super_positivity}, \ref{ass:super_ignorability} hold, the outcome is bounded with probability 1 (as in Theorem \ref{theorem:asymptotic_normality}) and the parametric form of the propensity score model indexed by $\vect \gamma$, $f_{\mb A|\mb L, i}(\mb a_i | \mb l_i;\vect\gamma)$, is correctly specified and differentiable with respect to $\vect\gamma$.
Let $\vect\mu_0(\alpha)$ be as in Theorem \ref{theorem:asymptotic_normality}, and $\widehat Y_N^L(a,\alpha)$ calculated using consistent estimates $\widehat{\vect\gamma}$ of the propensity score $f_{\mb A| \mb L, i}$. Let $\vect\psi_\gamma(\mb l_i, \mb a_i; \vect\gamma) = \frac{\partial}{\partial \vect\gamma^T} \log f(\mb a_i | \mb l_i;\vect\gamma)$ be the score functions. Assume that:
\begin{enumerate}
	\item $\vect\gamma_0$ is in an open subset of the Euclidean space
    \item $\vect\gamma \rightarrow \vect\psi_\gamma(\mb l_i, \mb a_i; \vect\gamma)$ is twice continuously differentiable $\forall (\mb l_i, \mb a_i)$
    \item $E_{F_0} \left\| \vect\psi_\gamma(\mb L_i, \mb A_i; \vect\gamma_0) \right\|^2_2 < \infty$ \label{cond:exp_norm_squared}
    \item $E_{F_0} \left[ \overset{\cdot}{\vect\psi}_\gamma(\mb L_i, \mb A_i; \vect\gamma_0) \right] $ exists and is non-singular \label{cond:exp_derivative}
    \item $\exists$ measurable integrable function $\overset{\cdot \cdot}{\psi}_\gamma(\mb l_i, \mb a_i)$ fixed such that $\overset{\cdot \cdot}{\psi}_\gamma$ dominates the second partial derivatives of $\vect\psi_\gamma$ for all $\vect \gamma$ in a neighborhood of $\vect\gamma_0$. \label{cond:dominate}
\end{enumerate}
where $\vect\gamma_0$ are the true parameters of the propensity score model, and $\overset{\cdot}{\vect\psi}_\gamma(\mb l_i, \mb a_i; \vect\gamma)$ is the matrix of partial derivatives of $\vect \psi_\gamma(\mb l_i, \mb a_i; \vect\gamma)$ with respect to $\vect\gamma$. Then,
\( \displaystyle
\sqrt{n}\left(\widehat Y_N^L(\alpha) - \vect\mu_0(\alpha) \right)
\overset{d}{\rightarrow} N(0, W(\vect\gamma_0, \vect\mu_0(\alpha))),
\)
where
\begin{align*}
& W(\vect\gamma_0, \vect\mu_0(\alpha)) = V(\vect\mu_0(\alpha)) + A_{21}B_{11}^{-1}A_{21}^T + A_{21}B_{11}^{-1}B_{12} +
\left(A_{21}B_{11}^{-1}B_{12}\right)^T , \\
& A_{21} =  E \Big[ \partial\psi_0/\partial\vect\gamma \ \ \partial\psi_1/\partial\vect\gamma \big]^T, \ 
B_{11} = E\Big[ \vect\psi_\gamma \vect\psi_\gamma^T\Big],\\
& B_{12} = E\Big[\vect\psi_\gamma \psi_0, \vect\psi_\gamma\psi_1\Big],
\end{align*}
evaluated at $(\vect\gamma_0, \vect\mu_0(\alpha))$,
$\psi_a = \psi_{a, \alpha}(\mb Y_i, \mb A_i, \mb L_i ; \mu_0(a, \alpha))$
and $V(\vect\mu_0(\alpha))$ is that of Theorem \ref{theorem:asymptotic_normality}.
\end{theorem}

$W(\vect\gamma_0,\vect\mu_0(\alpha))$ can be easily estimated using $\widehat W\left(\widehat{\vect\gamma}, \widehat Y_N^L(\alpha) \right)$, where $\widehat W\left(\vect\gamma, \vect\mu \right)$ is the matrix $W(\vect\gamma, \vect\mu)$ where all expectations are substituted with the empirical expectations. For example, $\widehat B_{11} = \frac1N \sum_{i = 1}^N \vect\psi_\gamma(\mb L_i, \mb A_i; \widehat{\vect\gamma}) \vect\psi_\gamma(\mb L_i, \mb A_i; \widehat{\vect\gamma})^T$.

%
Next, we derive the asymptotic distribution for $\widehat{\vect\mu}^{IE}(\alpha_0, \alpha_1) = \left(\widehat{Y}_N^L(0; \alpha_0), \ \widehat{Y}_N^L(0; \alpha_1) \right)^T$ for the estimated propensity score from a correctly specified parametric model.

\begin{theorem}
\label{theorem:asymptotic_normality_indirect}
If the assumptions of Theorem \ref{theorem:asymptotic_normality_estimated_ps} hold and for
$\vect\mu_0^{IE}(\alpha_0, \alpha_1) = \big(\mu_0(0, \alpha_0), $ $\mu_0(0, \alpha_2)\big)^T,$
\( \displaystyle
\sqrt{n}\left( \widehat{\vect\mu}^{IE}(\alpha_0, \alpha_1) - \vect\mu_0^{IE}(\alpha_0, \alpha_1) \right) \rightarrow N\big(0, Q(\vect\gamma_0, \vect\mu_0^{IE}(\alpha_0, \alpha_1))\big),
\)
where
\begin{align*}
& Q(\vect\gamma, \vect\mu) = D_{22} + C_{21}B_{11}^{-1}C_{21}^T + C_{21}B_{11}^{-1}D_{12} + \left(C_{21}B_{11}^{-1}D_{12}\right)^T \\
& D_{22} = Cov \bigg[ \Big(\overline Y_i^L(0, \alpha_1), \overline Y_i^L(0, \alpha_2) \Big)^T \bigg], \\ 
& D_{12} = E\left[\vect\psi_\gamma \psi_{0, \alpha_1}, \vect\psi_\gamma \psi_{0, \alpha_2}\right], \\
& C_{21} = E \big[ \partial\psi_{0, \alpha_1}/\partial\vect\gamma \ \ \partial\psi_{0, \alpha_2}/\partial\vect\gamma \big],
\end{align*}
and $B_{11}$ as in Theorem \ref{theorem:asymptotic_normality_estimated_ps}, evaluated at $(\vect\gamma_0, \vect\mu_0^{IE}(\alpha_1, \alpha_2))$.
\end{theorem}

\subsection{Estimators and asymptotic results for the population-level estimands}

Similar arguments lead to estimators of the $F_\alpha$-group and population average potential outcome in (\ref{eq:Falpha_estimand}) as
\begin{align*}
	\widehat{Y}_i^L(a ; F_\alpha) = & \int \widehat{Y}_i^L(a;\alpha) \mathrm{d} F_\alpha
    = \frac 1{n_i} \sum_{j = 1}^{n_i}
	\frac{\numeratorF}{\denominator} I(A_{ij} = a) Y_{ij}, \\
    \widehat{Y}^L(a, F_\alpha) = & \int \widehat Y^L(a;\alpha) \mathrm{d} F_\alpha = 
    \frac1N \sum_{i = 1}^{n_i} \widehat{Y}_i^L(a ; F_\alpha) 
\end{align*}
accordingly, where 
\begin{align*}
& \numeratorF = \int \numerator \ \mathrm{d} F_\alpha(\alpha).
\end{align*}

Assume that $F_\alpha^1, F_\alpha^2$ represent discrete distributions with values $\alpha_1, \alpha_2, \dots, \alpha_K \in (0, 1)$ and probability $p_{1k}$ and $p_{2k}$  of assigning value $\alpha_k$ to a cluster accordingly, such that
\(
\sum_{k = 1}^K p_{jk} = 1, \ j = 1, 2.
\)
Then,
\begin{align*}
\overline Y\left(0, F_\alpha^j \right) = \sum_{k=1}^K p_{jk} \overline Y(0, \alpha_k)\  \Rightarrow 
IE\left(F_\alpha^1, F_\alpha^2 \right) = \sum_{k = 1}^K (p_{2k} - p_{1k}) \overline Y(0,\alpha_k).
\end{align*}
Clearly, a consistent estimator for the super-population counterpart of the indirect effect $IE\left(F_\alpha^1, F_\alpha^2 \right) =
E_{F_0} \big[ \overline Y_i(a;F_\alpha^1)\big] - E_{F_0} \big[ \overline Y_i(a;F_\alpha^2)\big]$ is
\[
\widehat{IE} \left(F_\alpha^1, F_\alpha^2 \right) = \sum_{k = 1}^K (p_{2k} - p_{1k}) \widehat Y(0,\alpha_k).
\]
Acquiring the asymptotic distribution of $\widehat{IE} \left(F_\alpha^1, F_\alpha^2 \right)$ is straightforward following similar arguments to the ones in Theorem \ref{theorem:asymptotic_normality_indirect} to acquire the asymptotic distribution of $\big( \widehat Y(0,\alpha_1),$ $\widehat Y(0, \alpha_2), \dots, \widehat Y(0, \alpha_K) \big)^T$ and applying the multivariate delta method.


\section{Simulations}
\label{sec:simulations}

We generate a fixed population of 2,000 clusters including 14 to 18 units each, resulting to a total of 31,553 units.
Four independent $N(0, 1)$ covariates were generated, and are denoted as $L_1, L_2, L_3, L_4$. For every individual in the population (unit $j$ in cluster $i$), the potential outcomes under all possible treatment allocations were generated, following a model $Y \sim $ Bernoulli$(\mathrm{expit}(l_Y))$ where
\begin{align*}
l_Y =& 0.5 - 0.6a - 1.4 \frac{a + k}{n_i} - 0.098 L_{1ij} -0.145 L_{2ij} + 0.1 L_{3ij} +  0.3 L_{4ij} +  0.351 a \frac{a + k}{n_i},
\numberthis
\label{eq:sim_po}
\end{align*}
$L_{1ij}, L_{2ij}, L_{3ij}, L_{4ij}$ are the values of the covariates for observation $j$ of cluster $i$, $a$ is the individual treatment, $k$ is the number of treated neighbors, and $(a + k) / n_i$ is the percentage of units in the cluster that are treated.

\subsection{A simulated data set}

The simulations test the operating characteristics of the estimator in (\ref{eq:est_pop_aver_po}) using the true and estimated propensity score in terms of the re-sampling of the observed treatment vector. Specifically, each simulated dataset includes the whole population, but a different set of potential outcomes is observed according to a treatment vector generated as
$A_{ij}\sim$ Bernoulli$(\mathrm{expit}(l_A))$ where
\begin{align*}
l_A = & - 0.2 + b_i + 0.3 L_{1ij} - 0.15 L_{2ij} + 0.2 L_{3ij} -0.18 L_{4ij},\  b_i \sim N(0, 0.5^2).
\numberthis
\label{eq:sim_ps}
\end{align*}
Once the observed treatment is generated, the observed outcome is the corresponding value of the potential outcomes.

$\widehat Y_i(a;\alpha)$ is estimated from (\ref{eq:est_group_average_cond}) for $\counteralloc$ described in \cref{subsec:sim_Palpha}, and
\begin{align*}
\denominatorg = \int \prod_{j = 1}^{n_i}
f_e\left(A_{ij} | L_{ij}, \delta_0, \beta_i, \vect\delta \right) \phi \left(\beta_i | \sigma^2_\beta\right) \mathrm{d} \beta_i,
\end{align*}
where
\begin{align*}
f_e\left(A_{ij} | L_{ij}, \delta_0, \beta_i, \vect\delta \right) = & \ 
\mathrm{expit} \left(\delta_0 + b_i + L_{ij}^T\vect{\delta} \right)^{A_i} 
\left[ 1 - \mathrm{expit} \left( \delta_0 + b_i + L_{ij}^T\vect{\delta} \right) \right]^{1 - A_i},
\end{align*}
$L_{ij}^T = (L_{1ij}, L_{2ij}, L_{3ij}, L_{4ij})$, $\phi\left(\cdot;\sigma^2_\beta\right)$ the density of a 
$N\left(0, \sigma^2_\beta\right)$, and
$\vect\gamma = \left(\delta_0, \vect\delta, \sigma^2_\beta \right)$ known and equal to the coefficients in (\ref{eq:sim_ps}), or the maximum likelihood estimates from the correctly specified propensity score model.

We calculate the population average potential outcomes, direct and indirect effects, and the corresponding asymptotic variances.

\subsection{Covariate-dependent counterfactual treatment allocation}
\label{subsec:sim_Palpha}

The counterfactual treatment allocation $P_{\alpha, L}$ is allowed to depend on the same covariates that are included in the observed propensity score, using the log odds coefficients used to generate the observed treatment. Specifically, for a fixed $\alpha \in (0, 1)$,
\begin{align*}
& \mathrm{logit}P_{\alpha, L}(A_{ij} = 1|L_{ij}) = 
\xi_i^\alpha + 0.3 L_{1ij} - 0.15 L_{2ij} + 0.2 L_{3ij} -0.18 L_{4ij},
\end{align*}
for $\xi_i^\alpha$ satisfying $\frac1{n_i}\sum_{j = 1}^{n_i} P_{\alpha, L}(A_{ij} = 1|L_{ij}) = \alpha$. (Description of how $\xi_i^\alpha$ is calculated can be found in \cref{supp_sec:computation}.)


\subsection{Calculating the true average potential outcomes}
\label{subsec:sim_true_po}

For every observation $j$ in cluster $i$, the individual average potential outcome for individual treatment $a$ and for cluster-average propensity of treatment $\alpha$ is calculated based on (\ref{eq:ind_aver_po}).
%
Based on the individual average potential outcome, the true group and population average potential outcome are calculated according to (\ref{eq:group_aver_po}), (\ref{eq:pop_aver_po}).

\subsection{Simulation results}

We present results for values of $\alpha \in (0.25, 0.65)$ corresponding to the $10^{th}$ and $90^{th}$ quantiles of the distribution of the observed treatment proportions across clusters and simulated data sets. As expected, the estimator based on the true propensity score is unbiased, while the estimator based on the estimated propensity score, which is consistent but not unbiased, indicates small biases. Figure \ref{fig:sim_results} shows the mean estimate across 500 simulated data sets for the population average potential outcome for $a = 1$ (results were similar for $a = 0$), direct and indirect effect, whereas \cref{fig:sim_coverage} and \cref{tab:sim_coverage} depict the coverage of the estimators based on the true and estimated propensity score over different values of $\alpha$.

\begin{figure}
\centering
\includegraphics[width = \textwidth]{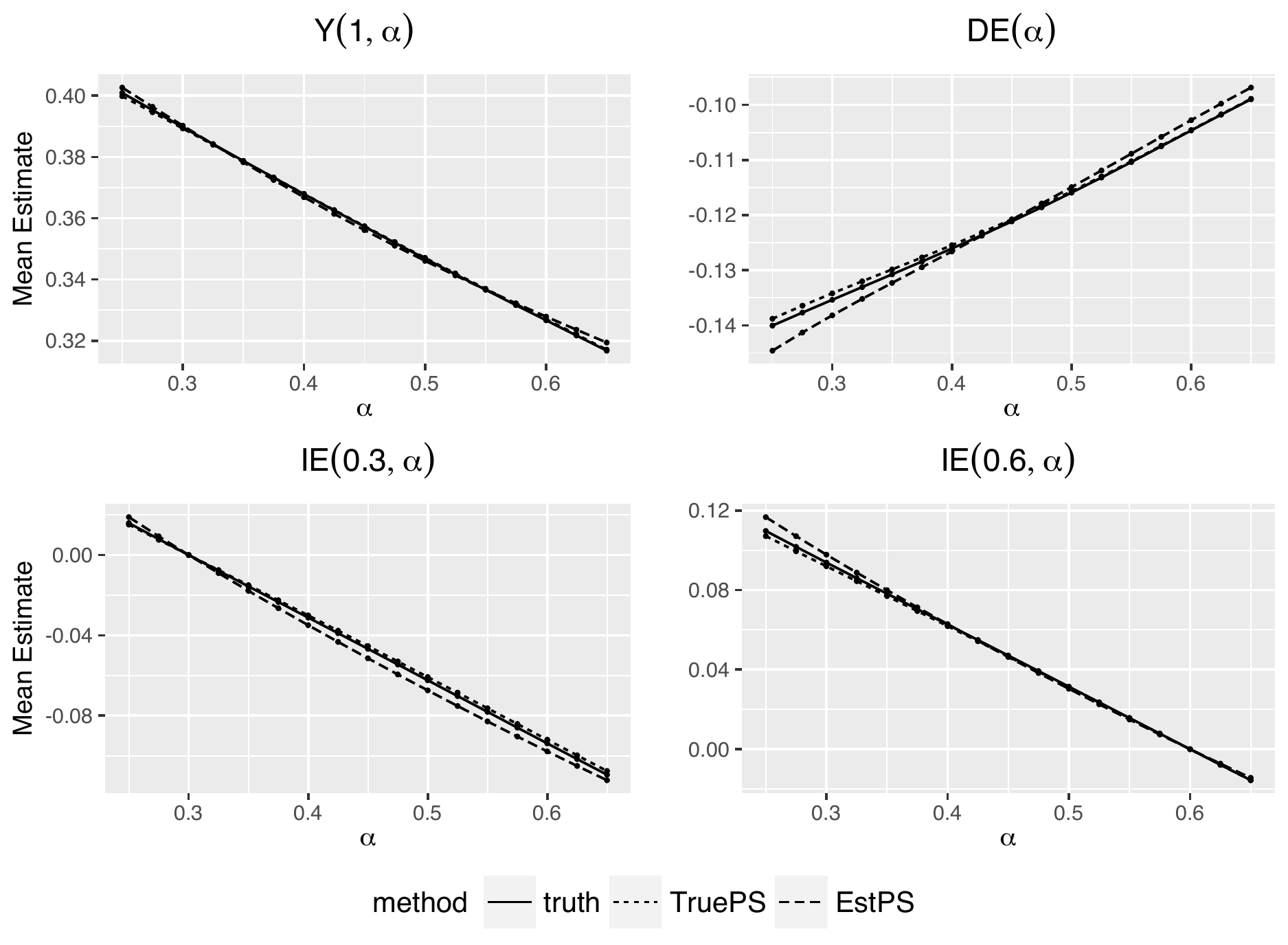}
\caption{Mean estimate of population average potential outcome, direct, and indirect effect over 500 simulated data sets for the true or the correctly specified propensity score.}
\label{fig:sim_results}
\end{figure}

\begin{figure}
\centering
\includegraphics[width = \textwidth]{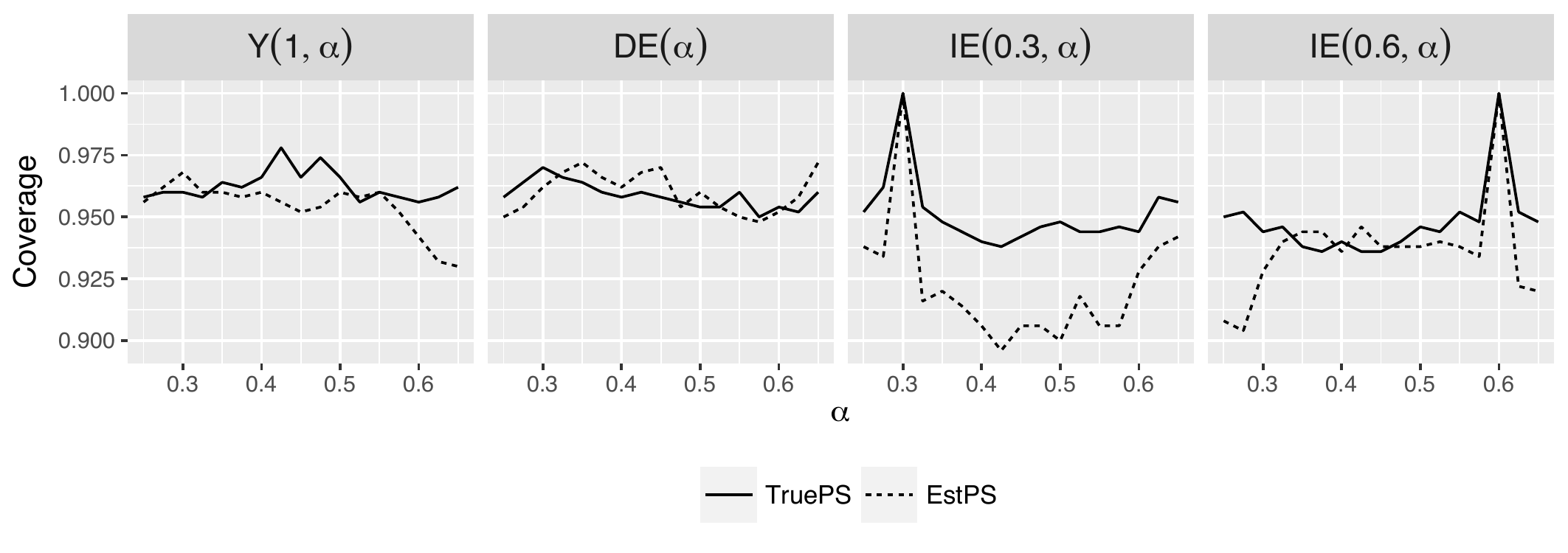}
\caption{Estimated coverage of population average potential outcome, direct and indirect effect over 500 simulated data sets for the true and the correctly specified propensity score model as a function of $\alpha$.}
\label{fig:sim_coverage}
\end{figure}

\begin{table}
	\centering
	\caption{Range of percent coverage over 500 simulated data sets for the population average potential outcome, direct and indirect effects for the true or the correctly specified propensity score.}
	\label{tab:sim_coverage}
	\begin{tabular}{cccccc}
		\hline
		PS & $\overline Y^L(0;\alpha)$ & $\overline Y^L(1;\alpha)$ & $DE^L(\alpha)$ & $IE^L(\alpha_1, \alpha_2)$ \\
		\hline
		True & 94.2 - 96.4 & 95.6 - 97.8 & 95 - 97 & \phantom{.0}93 - 96.2 \\
		Estimated & \phantom{.0}93 - 96.4 & \phantom{.0}93 - 96.8 & 94.8 - 97.2 & 88.2 - 95\phantom{.0} \\
		\hline
	\end{tabular}
\end{table}

Moreover, \cref{fig:sim_variance} in \cref{app_sec:sims} compares the mean of the estimated variance based on the asymptotic results against the variance of the estimates calculated over the 500 simulated data sets, indicating that, on average, the variance based on the asymptotic theory is a good approximation of the true variance.

\section{Application: Effectiveness of Power Plant Emissions Controls for Reducing Ambient Ozone Pollution}
\label{sec:data_application}

Limited literature exists in the evaluation of U.S. air pollution regulations in a causal inference framework. Power plant regulations for the reduction of NO$_x$ emissions have been predicated on the knowledge that reducing NO$_x$ emissions would lead to a subsequent reduction in ambient ozone. Among various NO$_x$ emission reduction strategies, SCR and SNCR are believed to be the most effective in reducing emissions. While work in \cite{Papadogeorgou2018} corroborated this effectiveness of SCR and SNCR in an analysis for NO$_x$ emissions, the analysis of ambient ozone pollution in that paper ignores the possibility of interference and estimates a null effect on ambient ozone.  However, interference is a key component in the study of air pollution: ambient pollution concentrations near a power plant will depend on the treatment levels of other nearby power plants. Causal estimands tailored to settings of interference can answer important questions related to the effectiveness of interventions in the presence of long-range pollution transport.  

We use the same data as in \cite{Papadogeorgou2018} to estimate direct and indirect effects of SCR/SNCR against alternatives on ambient ozone under realistic counterfactual programs. The publicly-available data set includes 473 coal or gas burning power generating facilities in the U.S. operating during June, July and August 2004, with covariate information on power plant characteristics, weather and demographic information of the surrounding areas. For every power plant, the value of ozone is calculated as the average across EPA monitoring locations within 100km of the 4$^{th}$ highest ozone measurements. See \cite{Papadogeorgou2018} for a full description of the data set and linkage.

Power plant facilities are grouped into 50 clusters according to Ward's agglomerative clustering method \citep{Ward1963} based on coordinates. The grouping and treatment of facilities are depicted in Figure \ref{fig:trt_clusters}.

\begin{figure}
\centering
\includegraphics[width = 0.75\textwidth]{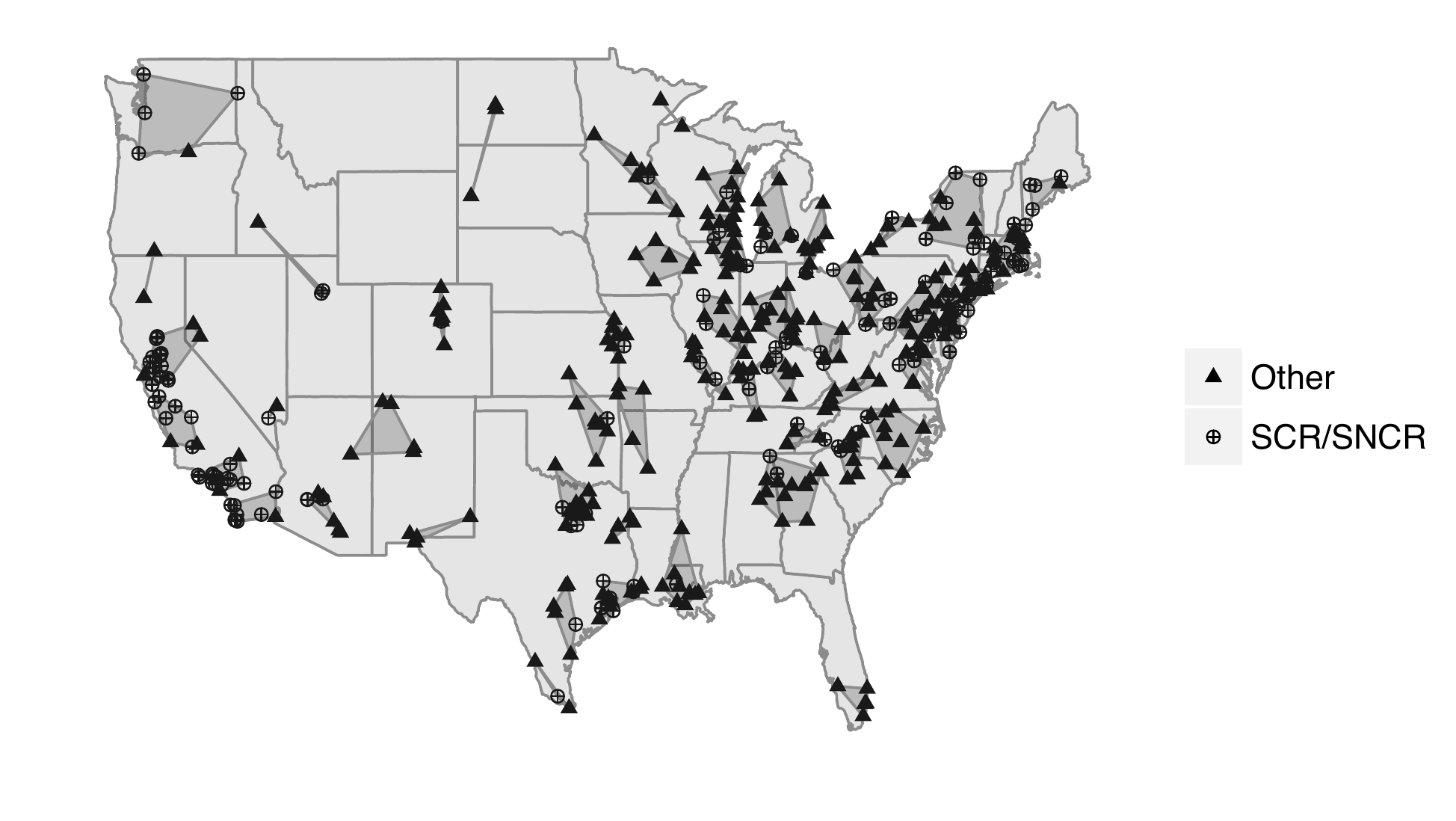}
\caption{Treated (SCR/SNCR) and control (Other) power plant facilities during June, July, August of 2004. Shaded areas depict the interference clusters according to the agglomerative clustering method.}
\label{fig:trt_clusters}
\end{figure}

\subsection{Plausibility of the ignorability and positivity assumption}

While regulatory programs provide incentives to install emission-control technologies, power plants have latitude to select which (if any) technology to adopt.  Such decisions are largely determined by the plant's characteristics such as plant size and operating capacity, as well as by factors related to local or regional air pollution incentives that are influenced by area-level characteristics such as population density and urbanicity.  To capture such factors, 18 covariates are included in the data set describing power plant, weather, and demographic characteristics, based on which ignorability is expected to hold. 
The variability in the observed proportion of treated power plants across clusters provides an additional indication that the positivity assumption is plausible.
Based on these covariates, the propensity score was modeled as in \cite{Papadogeorgou2018} augmented with a cluster-specific random effect
\begin{equation}
\mathrm{logit} P(A_{ij} = 1 | L_{ij}, b_i) = \delta_0 + b_i + L_{ij}^T \vect{\delta}, \ b_i \sim N(0, \sigma^2_b).
\label{eq:data_ps}
\end{equation}

\subsection{Counterfactual treatment allocation for the installation of SCR/SNCR emission control technologies}
Recall from \cref{sec:choose_numerator} that $\counteralloc$ governing treatment assignment in the counterfactual allocation programs of interest must be specified. To specify counterfactual treatment allocations that reflect realistic relationships between covariates and the propensity to adopt treatment, we specify $\counteralloc$ such that the log-odds of treatment installation related to individual covariates are as observed in the propensity score model for the observed treatment in (\ref{eq:data_ps}). Even though this choice of $P_{\alpha, L}$ depends on the data through the estimated log-odds, the corresponding estimands are well-defined and the asymptotic results are valid for $P_{\alpha, L}$ fixed across replications of the sampling or an increasing number of clusters.

Values of $\alpha$ were considered between the $20^{th}$ and $80^{th}$ quantiles of the observed cluster treatment proportions, corresponding to $\alpha \in [0.073, 0.458]$. Figure \ref{fig:app_results} shows the population direct effect $DE(\alpha)$, and population indirect effect $IE(\alpha_1, \alpha_2)$ for a subset of values of $\alpha_1$ (for presentation simplicity). The direct effect is significantly negative for all values of $\alpha \geq 0.12$, but has a somewhat increasing trend, implying that in a world where the average probability of SCR/SNCR among power plants in a cluster is fixed, the installation of SCR/SNCR at one power plant would lead to significant reductions in ozone concentrations in the surrounding area, but these reductions are smaller when the cluster average propensity of treatment is high (larger number of treated neighbors).

\begin{figure}
\centering
\includegraphics[width = \textwidth]{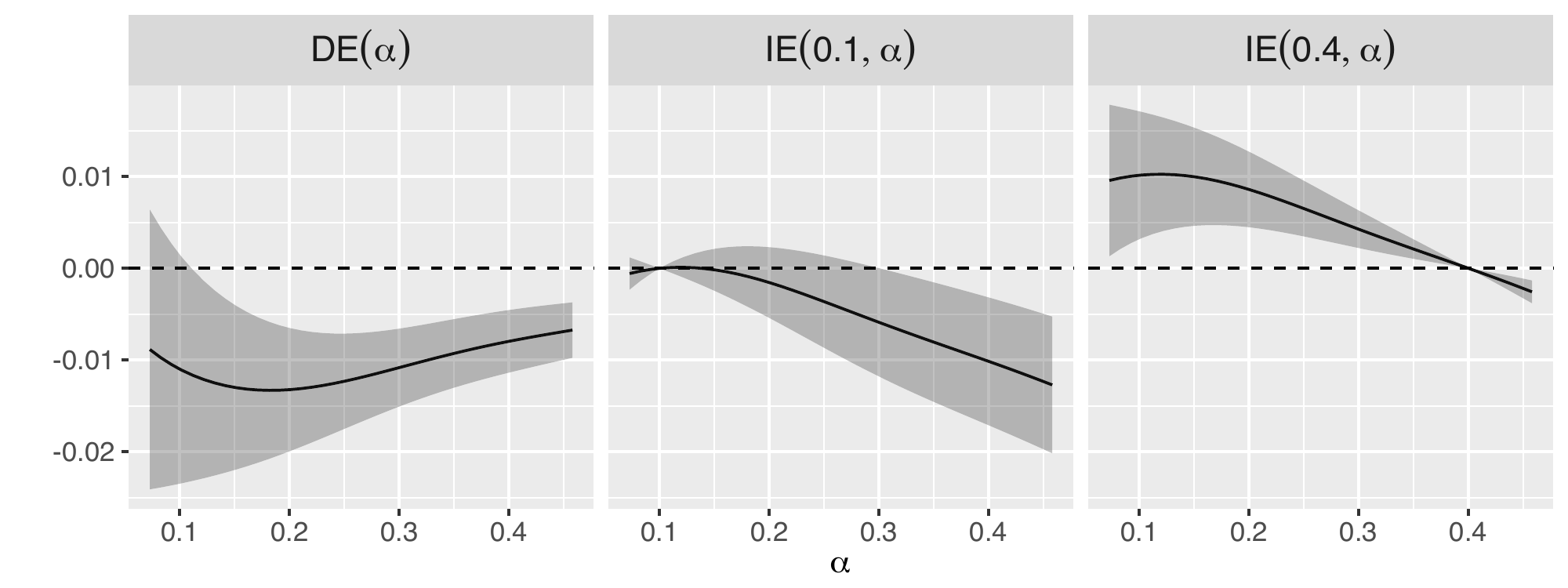}
\caption{Direct effect of control versus treated power plants on ozone concentrations as a function of $\alpha$, and indirect effect where the first value of $\alpha$ is fixed to a specific value. Ozone is measured in parts per million.}
\label{fig:app_results}
\end{figure}

The indirect effect is, in a way, a measure of pollution transport since it quantifies the effect of changes in the cluster average propensity of treatment on ozone concentrations near control power plants. For all values of $\alpha_1$, $IE(\alpha_1, \alpha_2)$ is almost always decreasing in $\alpha_2$, and most contrasts considered for which $\alpha_2 > 0.15$ were significant at the 0.05 significance level. The decreasing trend in $IE(\alpha_1, \alpha_2)$ for a fixed value of $\alpha_1$ implies that higher cluster-average SCR/SNCR propensity leads to further decrease in ambient ozone concentrations in the surrounding area of power plants without SCR/SNCR systems.

Next, we considered estimating the effect of hypothesized federal regulations that would shift the distribution of cluster-average propensity of treatment. $F_\alpha^1$ ($F_\alpha^2$) was assumed to be a discrete distribution within the $20^{th}$ ($50^{th}$) and $80^{th}$ quantiles of the observed cluster-treatment proportions. In Figure \ref{fig:app_Falphas}, we show the empirical probability mass function, as well as the two counterfactual treatment allocations. $IE\left(F_\alpha^1, F_\alpha^2 \right)$ was estimated to be $-0.0036$ parts per million (95\% CI: 
$-0.0059$ to $-0.0013$) implying that federal regulations that encourage the installation of SCR/SNCR
enough to bring the cluster average treatment propensities distribution from falling between the $20^{th}$ and $80^{th}$ percentiles of the observed cluster coverage distribution to falling between the $50^{th}$ and $80^{th}$ percentiles of the observed cluster coverage distribution, would lead to ambient ozone concentrations surrounding control power plants that are on average 0.0036 parts per million lower. For reference, these effect estimates can be compared against the national ozone air quality standard of 0.07 parts per million.

\begin{figure}
\centering
\includegraphics[width = \textwidth]{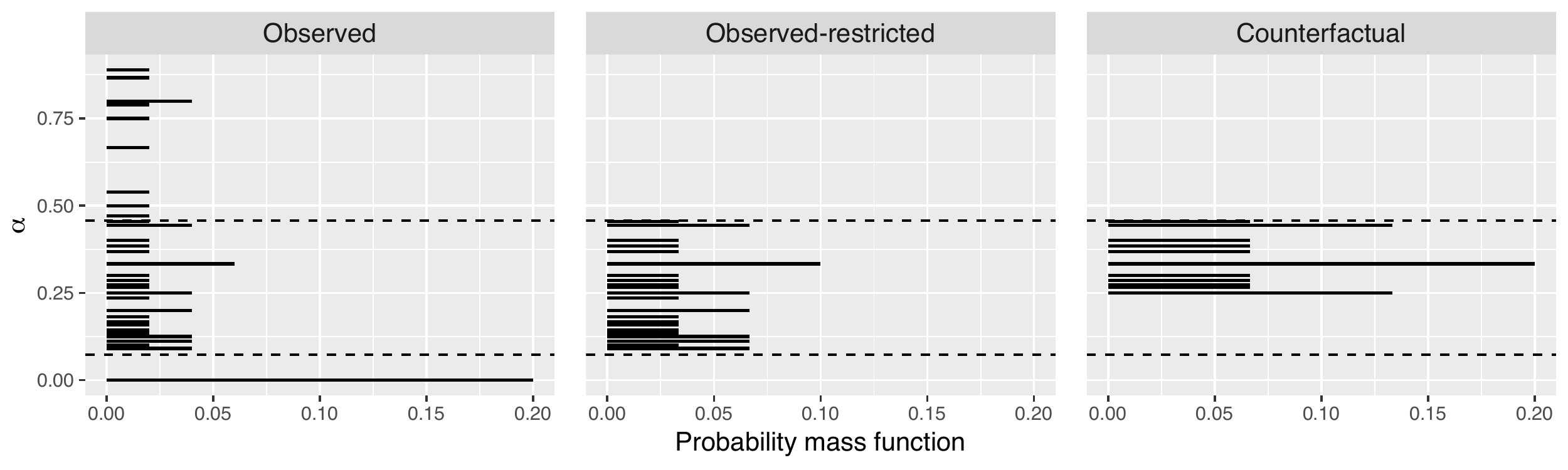}
\caption{Observed cluster treatment proportions (``Observed''), and two discrete hypothesized distributions of cluster-average probability of treatment. One corresponds to the observed restricted within the $20^{th}$ and $80^{th}$ quantiles of the observed cluster treatment proportions (``Observed-restricted''), and the other one (``Counterfactual'') corresponds to the observed (or the Observed-restricted) further restricted between the $50^{th}$ and $80^{th}$ quantiles of the observed cluster treatment proportions.}
\label{fig:app_Falphas}
\end{figure}

We explored the sensitivity of the results to the choice of hierarchical clustering method and number of clusters, and saw that the qualitative results for the effectiveness of SCR/SNCR emission reduction technologies are mostly consistent with negative estimated direct effects and decreasing indirect effect curves. 
These results can be found in \cref{app_sec:data_application}, along with links to the publicly available data set, R package and scripts.


\section{Discussion}
\label{sec:discussion}

Analyzing data in the context of interference disentangles the effect of the individual treatment from the treatment of one's neighbors. New estimands in the presence of interference were proposed for counterfactual strategies that manipulate treatment at the cluster-level, or at the level of population of clusters. These new estimands represent scenarios where individual treatment in the counterfactual world is allowed to depend on covariates and the treatment of one's neighbors. Such estimands are relevant for public health interventions that do not manipulate treatment at the unit level.

For the estimands referring to interventions at the population level, the counterfactual distribution $F_\alpha$ represented the distribution of the cluster-average propensity of treatment, and each cluster was assumed to be equally likely to receive $\alpha$ from $F_\alpha$. Alternative specifications could consider $F_\alpha$ to depend on cluster-level covariates that act as predictors of cluster-average propensity of treatment.  Further development could consider counterfactual treatment allocation strategies that manipulate the relationship between covariates and treatment assignment to reflect, for example, interventions for which larger power plants receive higher penalties for over-emission.

Consistent estimators were proposed for which the asymptotic distribution was derived. These estimators were employed in the comparative effectiveness of power plant emission control strategies on ambient ozone, and showed the potential of a set of emission reduction technologies in reducing ozone concentrations. These results are more in line with subject-matter knowledge than results from a previous study that assumed no interference.

While the power plant analysis showed the potential for causal inference methods for interference to lead to important results in air pollution research, there are several limitations worth noting. First of all, the number of clusters was low, raising questions for the appropriateness of use of asymptotic distributions to acquire variance estimates. Furthermore, the assumption of partial interference may be violated, since pollution from one power plant can travel long enough distances to affect ozone concentrations in a different cluster. Despite these approximations, the analysis of the air quality data entails important novelty in its own right, as it advances analysis methods for studies of air pollution interventions and introduces formalization of interference into a realm where it has not, to our knowledge, been previously considered. Further methods development, in particular towards relaxing the assumption of partial interference for unknown networks, is an important topic for future research.


\section*{Acknowledgements}
Funding for this work was provided by National Institutes of Health R01ES026217, USEPA 83587201-0, and Health Effects Institute 4953-RFA14-3/16-4. The contents of this work are solely the responsibility of the grantee and do not necessarily represent the official views of the USEPA. Further, USEPA does not endorse the purchase of any commercial products or services mentioned in the publication. The authors thank Dr. Christine Choirat for tools created to reproducibly manage and link the disparate (publicly-available) data used for this analysis.

\section*{Appendices}
\appendix

\section{Simulation results}
\label{app_sec:sims}
\begin{figure}[H]
	\centering
	\includegraphics[width = 0.8\textwidth]{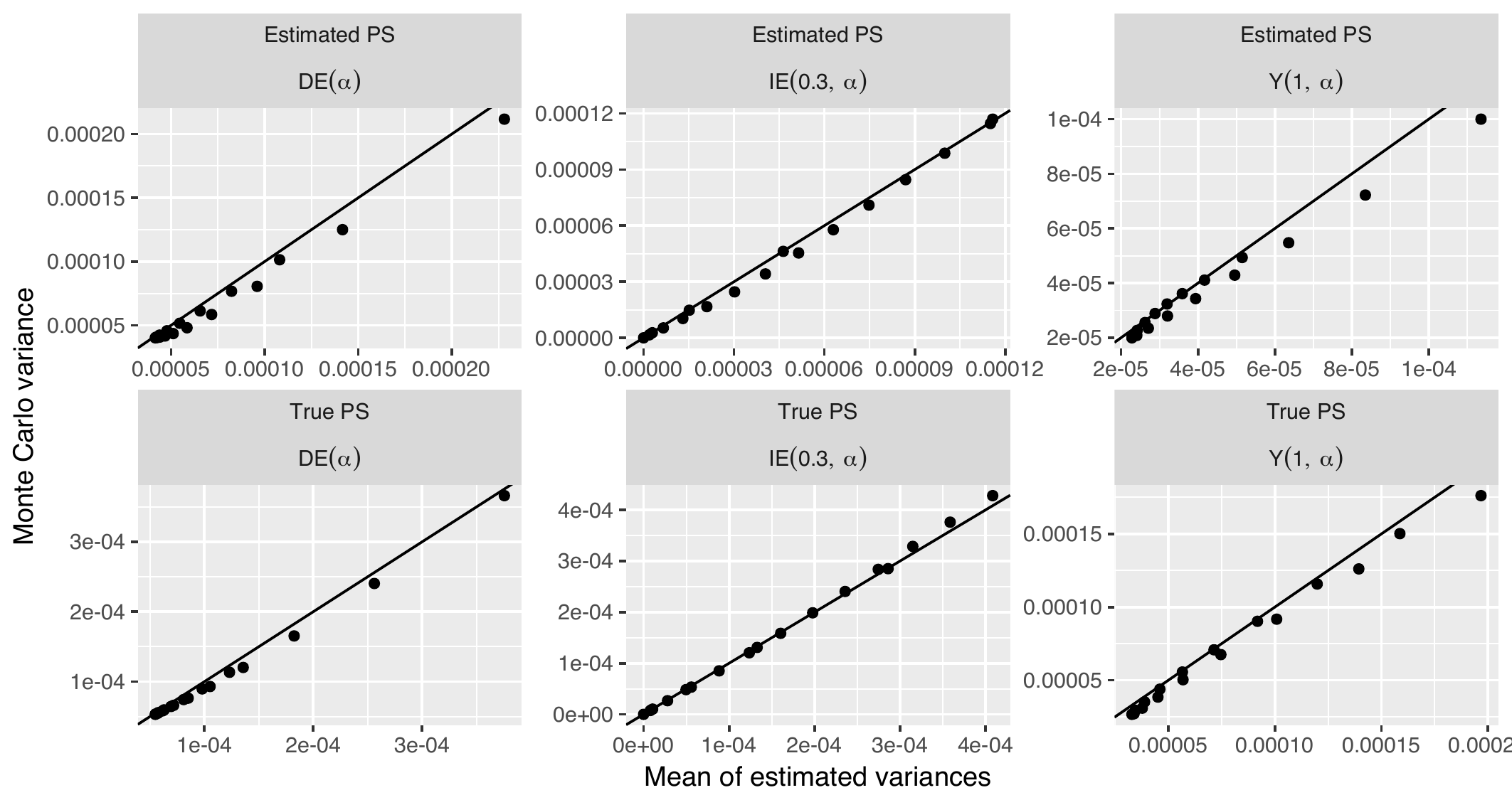}
	\caption{Mean estimated variance from the asymptotic distribution, and Monte Carlo variance of the estimates. The diagonal lines correspond to the 45 degree line, and each point corresponds to a value of $\alpha$.}
	\label{fig:sim_variance}
\end{figure}


\section{Data application}
\label{app_sec:data_application}

Link to the publicly available data, the R package implementing the estimators, and scripts replicating the results of the data analysis are available at 
\url{https://osf.io/7dp8c/} (page will be made public upon acceptance).

\subsection{Sensitivity of data application results to the choice of clustering}
\label{app_sec:app_sensitivity}

Rows  correspond to the direct effect $DE(\alpha)$ and the indirect effects $IE(\alpha_1, \alpha_2)$ for $\alpha_1 \in \{0.1, 0.4\}$. Columns correspond to the clustering method and correspond to Ward's \cite{Ward1963} method for 30 and 70 clusters, and complete clustering with 50 clusters. The decreasing trend in the indirect effect persists mostly for all clustering specifications, and the direct effect estimates are consistently negative.

\begin{figure}[H]
	\centering
	\includegraphics[width = \textwidth]{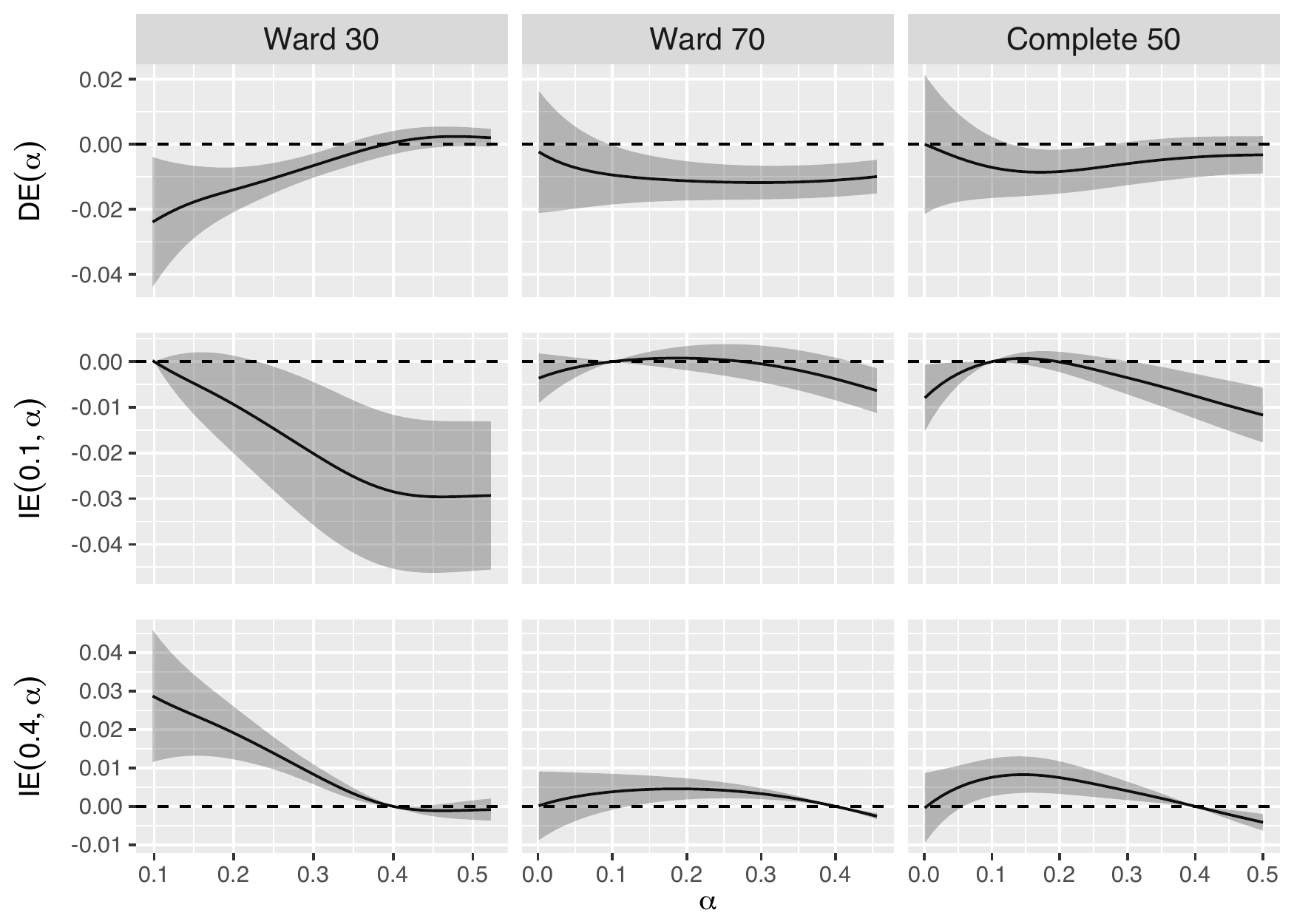}
	\caption{Direct and indirect effect of SCR/SNCR on ambient ozone using different clustering of power plants. Methods for clustering from left to right include Ward's method for 30 and 70 clusters, and complete clustering using 50 clusters.}
\end{figure}





\section{Proofs of unbiasedness, consistency and asymptotic normality}
\label{supp_sec:proofs}

\subsection{Unbiasedness}

\begin{theorem}
	If $f_{\mb A|\mb L, i}(\cdot | \mb L_i)$ is known, and Assumptions \ref{ass:group_positivity}, \ref{ass:group_ignorability} hold, then $\widehat{Y}^L_i(a ; \alpha)$ is an unbiased estimator for the group average potential outcome, and $\widehat{Y}^L(a;\alpha)$ is an unbiased estimator of the population average potential outcome for individual treatment $a$ and cluster average propensity of treatment $\alpha$.
	\label{theorem:unbiased_cond_alpha}
\end{theorem}

\begin{proof}
	All expectations are taken with respect to the conditional distribution $\mb A_i | \mb L_i, \mb Y_i(\cdot)$, where $\mb Y_i(\cdot)$ are all the potential outcomes for all units in cluster $i$.
	$Y_ij$, $\mb Y_i$ denote the observed individual outcome, and the vector of observed outcomes in cluster $i$ accordingly.
	\small{
	\begin{align*}
	& E[\widehat{Y}_i^L(a; \alpha)] \\
	= & \frac 1{n_i} \sum_{j = 1}^{n_i}
	E \left(
	\frac{f_{\mb A| \mb L, i, \alpha}(\mb A_{i, -j} | A_{ij} = a, \mb L_i, \alpha)}
	{f_{\mb A| \mb L, i}(\mb A_i | \mb L_i)} I(A_{ij} = a) Y_{ij}
	\right) \\
	= & \frac 1{n_i} \sum_{j = 1}^{n_i}
	E \left(
	\frac{f_{\mb A| \mb L, i, \alpha}(\mb A_{i, -j} | A_{ij} = a, \mb L_i, \alpha)}
	{f_{\mb A| \mb L, i}(\mb A_i | \mb L_i)} I(A_{ij} = a) Y_{ij}(\mb A_i)
	\right) \\
	= & \frac 1{n_i} \sum_{j = 1}^{n_i}
	\sum_{\mb s \in \mathcal{A}(n_i)}
	\frac{f_{\mb A| \mb L, i, \alpha}(\mb A_{i, -j} = \mb s_{i, -j} | A_{ij} = a, \mb L_i, \alpha)}
	{f_{\mb A| \mb L, i}(\mb A_i = \mb s | \mb L_i)}
	I(s_{ij} = a) Y_{ij}(\mb s)
	P(\mb A_i = \mb s | \mb L_i, \mb Y_i(\cdot)) \\
	= & \frac 1{n_i} \sum_{j = 1}^{n_i}
	\sum_{\mb s \in \mathcal{A}(n_i)}
	f_{\mb A| \mb L, i, \alpha}(\mb A_{i, -j} = \mb s_{i, -j} | A_{ij} = a, \mb L_i, \alpha)
	I(s_{ij} = a) Y_{ij}(\mb s)
	\tag{From Assumption \ref{ass:group_ignorability}
		$ P(\mb A_i = \mb s | \mb L_i, \mb{Y}_i(\cdot)) = P(\mb A_i = \mb s | \mb L_i) = f_{\mb A| \mb L, i}(\mb A_i = \mb s | \mb L_i)$.} \\
	= &
	\frac 1{n_i} \sum_{j = 1}^{n_i}
	\sum_{\mb s \in \mathcal{A}(n_i - 1)}
	f_{\mb A| \mb L, i, \alpha}(\mb A_{i, -j} = 	\mb s | A_{ij} = a, \mb L_i, \alpha)
	Y_{ij}(a_{ij} = a, a_{i, -j} = \mb s)
	= \overline{Y}_i^L (a; \alpha).
	\end{align*}}
	By linearity of expectations, the proof for the population average potential outcome is trivial.
\end{proof}


\subsection{Proofs of asymptotic results for known propensity score}
\label{app_subsec:ps_known}

For notational simplicity, denote $\tilde{\mb O}_i = (\mb A_i, \mb L_i)$, $\mb O_i = (\mb Y_i, \mb A_i, \mb L_i)$, $\tilde{\mb o}_i = (\mb a_i, \mb l_i)$, and $\mb o_i = (\mb y_i, \mb a_i, \mb l_i)$. Also, denote as $F_0$ the distribution of $(\mb{Y}_i(\cdot), \mb A_i, \mb L_i)$ in the superpopulation.

Consider the estimating equation $\Psi_N(\mu) = \sum_{i =1}^N \psi_{a, \alpha}(\mb O_i; \mu) = 0$, where
$$
\psi_{a, \alpha}(\mb O_i; \mu) =
\left( \frac1{n_i} \sum_{j = 1}^{n_i}
\frac{f_{\mb A| \mb L, i, \alpha}(\mb A_{i, -j} | A_{ij} = a, \mb L_i, \alpha)}
{f_{\mb A| \mb L, i}(\mb A_i | \mb L_i)} I(A_{ij} = a) Y_{ij}\right) - \mu.
$$
It is easy to see that the solution to this equation is $\widehat\mu = \widehat{Y}^L_N(a; \alpha)$:
\begin{align*}
& \sum_{i = 1}^N \left[ \left( \frac{1}{n_i} \sum_{j = 1}^{n_i}
\frac{f_{\mb A| \mb L, i, \alpha}(\mb A_{i, -j} |
	A_{ij} = a, \mb L_i, \alpha)}
{f_{\mb A| \mb L, i}(\mb A_i | \mb L_i)}
I(A_{ij} = a) Y_{ij} \right) - \mu \right] = 0 \iff \\
& \sum_{i = 1}^N \widehat{Y}_i^L(a; \alpha) = N \mu \iff 
\widehat \mu = \frac1N \sum_{i = 1}^N \widehat{Y}_i^L(a; \alpha)
= \widehat{Y}_N^L(a; \alpha)
\end{align*}
If $\mu_0 = \mu_0(a, \alpha)$ is the solution to    
$\Psi_0(\mu) = \int \psi_{a, \alpha}(\mb O_i; \mu) \mathrm{d}F_0(\mb o_i) = 0$. Then,
\begin{align*}
& \int \left[ \frac1{n_i}\sum_{j = 1}^{n_i} 
\frac{f_{\mb A| \mb L, i, \alpha}(\mb A_{i, -j} |
	A_{ij} = a, \mb L_i, \alpha)}
{f_{\mb A| \mb L, i}(\mb A_i | \mb L_i)}
I(A_{ij} = a) Y_{ij} - \mu_0(a, \alpha) \right] \ 
\mathrm{d}F_0(\mb o_i) = 0 \iff \\
&\mu_0(a, \alpha) = E_{F_0}\left[\frac1{n_i}\sum_{j = 1}^{n_i} 
\frac{f_{\mb A| \mb L, i, \alpha}(\mb A_{i, -j} |
	A_{ij} = a, \mb L_i, \alpha)}
{f_{\mb A| \mb L, i}(\mb A_i | \mb L_i)}
I(A_{ij} = a) Y_{ij} \right]
= E_{F_0}\left[\overline{Y}_i^L(a; \alpha)\right]
\end{align*}
%


\renewcommand*{\proofname}{\textbf{Proof of Theorem \ref{theorem:asymptotic_normality}}}
\begin{proof}
	First, we will show that $\widehat Y_N^L(\alpha)$ is consistent for $\boldsymbol\mu_0(\alpha)$. For this proof, we use an alteration of Lemma A in section 7.2.1 of \cite{Serfling1980}. \\ [5pt]
	Note that $\psi_{a, \alpha}(\mb O_i; \mu)$ is monotone in $\mu$ with
	$ \overset{\cdot}{\psi}_{a, \alpha}(\mb O_i;\mu)
	= \frac{\partial}{\partial \mu} {\psi}_{a, \alpha}(\mb O_i;\mu)
	= - 1 < 0 $.
	Therefore, $\Psi_N(\mu), \Psi_0(\mu)$ are also monotone in $\mu$ (implying uniqueness of their roots). From the strong law of large numbers we have that $\Psi_N(\mu) \overset{a.s.}{\rightarrow} \Psi_0(\mu)$. From this, we have that:
	$$
	|\Psi_0(\widehat\mu) - \Psi_0(\mu_0)| =
	|\Psi_0(\widehat\mu) - \Psi_n(\widehat\mu)| \leq
	\sup_\mu |\Psi_0(\mu) - \Psi_n(\mu)| \rightarrow 0,
	$$
	which, by the uniqueness of the roots for $\Psi_0, \Psi_N$, implies 
	$\widehat{Y}_N^L(a;\alpha) \overset{a.s.}{\rightarrow} E_{F_0}\left[\overline{Y}_i^L(a;\alpha)\right]$.
	
	From basic probability laws we have that since the individual components converge almost surely to their limit,
	$$\left(\widehat Y_N^L(0, \alpha), \widehat Y_N^L(1;\alpha)\right)^T
	\overset{a.s}{\rightarrow} \vect\mu_0(\alpha) =
	\left( E_{F_0}\left[ \bar Y_i^L(0, \alpha)\right],
	E_{F_0}\left[\bar Y_i^L(1, \alpha)\right]\right)^T,$$
	which also establishes convergence in probability.

	Now we will show that $\widehat Y_N^L(a;\alpha)$ has an asymptotically univariate normal distribution, for $a = 0, 1$, and afterwards we extend this to showing that  $\widehat Y_N^L(\alpha)$ has an asymptotically bivariate normal distribution. \\ [10pt]
	\textbf{Univariate result} \\
	Based on the above, $\widehat Y^L(a, \alpha) \overset{p}{\rightarrow} E_{F_0} \left[\overline Y_i^L(a, \alpha)\right] = \mu_0(a, \alpha)$. Theorem A in section 7.2.2 of \cite{Serfling1980} requires:
	\begin{enumerate}
		\item[(i)] $\mu_0(a, \alpha)$ is an isolated root of $\Psi_0(\mu) =0$ and $\psi_{a, \alpha}(\cdot;\mu)$ is monotone in $\mu$. (Shown above)
		\item[(ii)] $\Psi_0(\mu)$ is differentiable at $\mu_0(a, \alpha)$ with $\Psi_0'(\mu_0(a, \alpha)) \neq 0$.
		\item[(iii)] $\int \psi^2_{a, \alpha}(\mb o_i;\mu) \mathrm{d}F_0(\mb o_i)$ is finite in a neighborhood of $\mu_0(a, \alpha)$.
	\end{enumerate}
	\noindent
	\textit{Proof of (ii)}
	\begin{align*}
	\Psi_0(\mu) = &
	\int \left[ \frac1{n_i} \sum_{j = 1}^{n_i}
	\frac{f_{\mb A| \mb L, i, \alpha}(\mb A_{i, -j} | A_{ij} = a, \mb L_i, \alpha)}
	{f_{\mb A| \mb L, i}(\mb A_i | \mb L_i)}
	I(A_{ij} = a) Y_{ij} - \mu \right]
	\mathrm{d}F_0(\mb o_i) \\
	= &
	\frac1{n_i} \int \sum_{j = 1}^{n_i}
	\frac{f_{\mb A| \mb L, i, \alpha}(\mb A_{i, -j} | A_{ij} = a, \mb L_i, \alpha)}
	{f_{\mb A| \mb L, i}(\mb A_i | \mb L_i)}
	I(A_{ij} = a) Y_{ij} \ 
	\mathrm{d}F_0(\mb o_i)
	- \mu
	\end{align*}
	So $\Psi_0$ is linear in $\mu$ and therefore differentiable everywhere, with $\Psi'_0(\mu) = - 1 \neq 0$.
	
	\noindent
	\textit{Proof of (iii)}\\
	Consider a neighborhood of $\mu_0 = \mu_0(a, \alpha)$ of the form $(\mu_0(a, \alpha) - \epsilon, \mu_0(a, \alpha) + \epsilon)$, for some $\epsilon > 0$. Then,
	\begin{align*}
	& \int \psi^2_{a, \alpha}(\mb O_i;\mu) \mathrm{d}F_0(\mb o_i) \\
	= & \int \left( \frac1{n_i} \sum_{j = 1}^{n_i}
	\frac{f_{\mb A| \mb L, i, \alpha}(\mb A_{i, -j} | A_{ij} = a, \mb L_i, \alpha)}
	{f_{\mb A| \mb L, i}(\mb A_i | \mb L_i)}
	I(A_{ij} = a) Y_{ij} - \mu \right) ^2 \mathrm{d}F_0(\mb o_i) \\
	= & \int \left| \frac1{n_i} \sum_{j = 1}^{n_i}
	\frac{f_{\mb A| \mb L, i, \alpha}(\mb A_{i, -j} | A_{ij} = a, \mb L_i, \alpha)}
	{f_{\mb A| \mb L, i}(\mb A_i | \mb L_i)}
	I(A_{ij} = a) Y_{ij} - \mu \right| ^2 \mathrm{d}F_0(\mb o_i)
	\end{align*}
	I will show that
	$\left|\frac1{n_i} \sum_{j = 1}^{n_i}
	\frac{f_{\mb A| \mb L, i, \alpha}(\mb A_{i, -j} | A_{ij} = a, \mb L_i, \alpha)}
	{f_{\mb A| \mb L, i}(\mb A_i | \mb L_i)}
	I(A_{ij} = a) Y_{ij} - \mu \right|$
	is bounded by a constant $c$ in a neighborhood of $\mu_0$ and therefore the integral is bounded by $c^2$.
	\begin{align*}
	& \left| \frac1{n_i} \sum_{j = 1}^{n_i}
	\frac{f_{\mb A| \mb L, i, \alpha}(\mb A_{i, -j} | A_{ij} = a, \mb L_i, \alpha)}
	{f_{\mb A| \mb L, i}(\mb A_i | \mb L_i)}
	I(A_{ij} = a) Y_{ij} - \mu \right| \\
	& \leq 
	\frac1{n_i} \sum_{j = 1}^{n_i} \frac{f_{\mb A| \mb L, i, \alpha}(\mb A_{i, -j} | A_{ij} = a, \mb L_i, \alpha)}
	{f_{\mb A| \mb L, i}(\mb A_i | \mb L_i)} |Y_{ij}| + |\mu| \\
	& \leq \frac1{n_i} \sum_{j = 1}^{n_i} \frac{|Y_{ij}|}{\denominator} + |\mu| \\
	\Rightarrow & \int \psi^2_{a, \alpha}(\mb O_i;\mu) \mathrm{d}F_0(\mb o_i) \leq
	|\mu| + (n_i)^{-1} \sum_{j = 1}^{n_i}E_{F_0} \left[\frac{|Y_{ij}|}{\denominator} \right] \\
	& \leq
	|\mu| + (n_i)^{-1} \sum_{j = 1}^{n_i} \sqrt{E_{F_0}(Y_{ij}^2) E_{F_0}(\denominator^{-2})} <
	|\mu| + M \rho^{-1} = c
	\numberthis
	\label{eq:proof_bound_c}
	\end{align*}
	We have shown that the conditions of Theorem A (section 7.2.2 of \cite{Serfling1980}) are satisfied, and therefore
	$$
	\sqrt{n} \left( \widehat{Y}^L_N(a; \alpha) - \mu_0(a;\alpha) \right) \overset{d}{\rightarrow} N(0, \sigma^2),
	$$
	where $\sigma^2 =
	E\left[\psi^2_{a, \alpha}\left(\mb O_i; \mu_0(a, \alpha)\right)\right]$, since $\Psi_0'(\mu_0(a, \alpha)) = -1$.
	
	\vspace{20pt} \noindent
	\textbf{Bivariate result} \\
	We will use Theorem 5.41 of \cite{Vaart1998}. The assumptions of this theorem are the so-called ``classical'' conditions, and are stricter than necessary to prove asymptotic normality. However, this theorem is often used in practice, since the conditions are sometimes easy to prove, as they are here.
	
	We denote
	$ \psi(\mb o_i;\vect \mu) =
	(\psi_{0, \alpha}(\mb o_i; \mu^0),
	\psi_{1, \alpha}(\mb o_i; \mu^1))^T
	$, for $\vect\mu = (\mu^0, \mu^1)$,
	and $\Psi_n(\vect \mu)$, $\Psi_0(\vect\mu)$ similarly as above, but for the vector $\psi$.
	
	It was shown that $\vect \mu_0(\alpha)$ satisfies $\Psi_0(\vect{\mu}) = 0$, and that $\widehat{Y}_N^L(\alpha)$ is a consistent estimator of $\vect \mu_0(\alpha)$. In order to apply Theorem 5.41, we show that
	\begin{enumerate}
		\item[(i)] the function $\vect \mu \rightarrow \psi(\mb o_i;\vect \mu)$ is twice continuously differentiable for every vector $\mb o_i$,
		\item[(ii)] $E_{F_0}\| \psi(\mb O_i; \vect \mu_0(\alpha))\|_2^2 < \infty$ (where $\|\cdot\|_2$ is the 2-norm $\|(v_1, v_2, \dots, v_n)\|_2 = (v_1^2 + v_2^2 + \dots + v_n^2)^{1/2}$,
		\item[(iii)]  The matrix $E_{F_0} \left[\overset{\cdot}{\psi}(\mb O_i; \vect \mu_0(\alpha))\right]$ exists and is nonsingular, and
		\item[(iv)]  $\exists$ fixed integrable function $\overset{\cdot\cdot}{\psi}(\mb o_i)$ such that $\overset{\cdot\cdot}{\psi}$ dominates the second order partial derivatives of $\psi$ $\forall \vect \mu$ in a neighborhood of $\vect \mu_0(\alpha)$. 
	\end{enumerate}
	
	\noindent
	\textit{Proof of (i).}
	It has already been shown that $\psi_{a, \alpha}(\mb o_i;\mu)$ is linear in $\mu$ and therefore twice continuously differentiable with respect to $\mu$ for every vector $(\mb o_i)$.
	
	\noindent
	\textit{Proof of (ii).}
	\small{
	\begin{align*}
	& E_{F_0} \| \psi(\mb O_i; \vect \mu_0(\alpha))\|_2^2 \\
	=&  E\left\{ \sum_{a \in \{0, 1\}} \left[ \frac1{n_i} \sum_{j = 1}^{n_i}
	\frac{\numerator}{\denominator} I(A_{ij} = a) Y_{ij} - \mu_0(a, \alpha) \right]^2 \right\} \\
	=& \sum_{a \in \{0, 1\}} E \left[\frac1{n_i} \sum_{j = 1}^{n_i}
	\frac{f_{\mb A| \mb L, i, \alpha}(\mb A_{i, -j} | A_{ij} = a, \mb L_i, \alpha)}
	{f_{\mb A| \mb L, i}(\mb A_i | \mb L_i)} I(A_{ij} = a) Y_{ij} - \mu_0(a, \alpha) \right]^2 \\
	\leq & 2c^2 \tag{because of (\ref{eq:proof_bound_c})}
	\end{align*}}

	\noindent
	\textit{Proof of (iii).} 
	\begin{align*}
	\overset{\cdot}{\psi}(\mb o_i; \vect \mu) = &
	\begin{pmatrix}
	\frac{\partial \psi_{0, \alpha}\left(\mb o_i; \mu^0\right)}
	{\partial \mu^0} &
	\frac{\partial \psi_{0, \alpha}\left(\mb o_i; \mu^0\right)}
	{\partial \mu^1} \\
	\frac{\partial \psi_{1, \alpha}\left(\mb o_i; \mu^1\right)}
	{\partial \mu^0} &
	\frac{\partial \psi_{1, \alpha}\left(\mb o_i; \mu^1\right)}
	{\partial \mu^1}
	\end{pmatrix}
	= \begin{pmatrix} - 1 & 0 \\ 0 & - 1 \end{pmatrix} = - I_2 < \infty \text{ and non-singular},
	\numberthis
	\label{eq:proof_psi_dot}
	\end{align*}
	where the diagonal elements of partial derivatives are calculated in the proof of consistency, and the non-diagonal elements are clearly 0 since the functions do not include the corresponding components of $\vect \mu$.
	
	\noindent
	\textit{Proof of (iv).} Based on equation (\ref{eq:proof_psi_dot}), we have that all second order derivatives are equal to 0, and are therefore dominated by the integrable function $\overset{\cdot \cdot}{\psi}(\mb o_i) = 0$.
	
	\noindent
	From Theorem 5.41 of \cite{Vaart1998}, we have that
	\begin{align*}
	& \sqrt{n} \left( \widehat Y_N^L(\alpha) - \vect\mu_0(\alpha) \right) = - \left(E\left[\overset{\cdot}{\psi}(\mb O_i; \vect\mu_0(\alpha))\right]\right)^{-1} \frac1{\sqrt n} \sum_{i = 1}^N \psi(\mb O_i; \vect\mu_0(\alpha)) + o_P(1) \\
	\Rightarrow &
	\ \sqrt{n} \left( \widehat Y_N^L(\alpha) - \vect\mu_0(\alpha) \right) \overset{d}{\rightarrow}
	N\left(0, A(\vect\mu_0(\alpha))^{-1} V(\vect\mu_0(\alpha)) \left[ A(\vect\mu_0(\alpha))^{-1} \right]^T\right),
	\end{align*}
	where $$A(\vect\mu_0(\alpha)) = E\left[- \overset{\cdot}{\psi}(\mb O_i; \vect\mu_0(\alpha))\right] = I_2$$
	(from \ref{eq:proof_psi_dot}), and $$V(\vect \mu_0(\alpha)) = E\left[\psi(\mb O_i; \vect\mu_0(\alpha))\psi(\mb O_i; \vect\mu_0(\alpha))^T\right].$$
\end{proof}

\begin{example}
	We provide an example of the application of the delta method on the result of Theorem \ref{theorem:asymptotic_normality}. Consider the direct effect defined as $\mu_0^{DE}(\alpha) = \mu_0(1,\alpha) - \mu_0(0, \alpha)$. Then $\widehat \mu^{DE}(\alpha) = \widehat{Y}^L(1; \alpha) - \widehat{Y}^L(0, \alpha)$ is a consistent estimator, and can be written as $g(\widehat{Y}^L(\alpha))$ for $g((x_1, x_2)^T) = x_1 - x_2$.
	From the Delta method, we know that
	\[
	\sqrt{n} \left( \widehat \mu^{DE}(\alpha) - \mu_0^{DE}(\alpha) \right)
	\rightarrow N(0, \sigma^2)
	\]
	for $ \sigma^2 = \nabla g(\vect \mu_0(\alpha))^T V(\vect\mu_0(\alpha)) \nabla g(\vect \mu_0(\alpha))$, where $\nabla g((x_1, x_2)^T) = (\frac{\partial g}{\partial x_1}, \frac{\partial g}{\partial x_2})^T = (1, -1)^T,$ and $V(\vect\mu_0(\alpha))$ is as in Theorem \ref{theorem:asymptotic_normality}.
	\label{app_ex:delta_method}
\end{example}

\subsection{Proofs of asymptotic results for correctly specified propensity score}
\label{app_subsec:ps_est}

\begin{lemma}
	If condition \ref{cond:exp_norm_squared} of Theorem \ref{theorem:asymptotic_normality_estimated_ps} holds, then
	$ E\left[\vect\psi_\gamma(\mb L_i, \mb A_i; \vect\gamma_0)\right] < \infty $
	\label{lemma:exp_score}
\end{lemma}
%
\renewcommand*{\proofname}{\textbf{Proof of Lemma \ref{lemma:exp_score}}}
\begin{proof}
	Denote $\vect\psi_\gamma = (\psi_\gamma^1, \psi_\gamma^2, \dots, \psi_\gamma^p)^T$. Then,
	$$
	E_{F_0}^2\left(\psi_\gamma^k\right) \leq
	E\left[\left(\psi_\gamma^k\right)^2\right] \leq
	\sum_{l = 1}^p E\left[\left(\psi_\gamma^l\right)^2\right] =
	E_{F_0}^2 \left\| \vect\psi_\gamma(\mb L_i, \mb A_i;\vect\gamma) \right\|^2 < \infty \Rightarrow
	E_{F_0}\left(\psi_\gamma^k\right) < \infty
	$$
	where the first inequality uses Jensen's inequality for $g(x) = x^2$. From this, we see that the score functions are integrable with finite expectation.
\end{proof}

\begin{lemma}
	Assuming that the conditions of Theorem \ref{theorem:asymptotic_normality_estimated_ps} hold, the estimator $\widehat{Y}_N^L(\alpha)$ using the estimates of the correctly specified propensity score model is consistent for $\vect\mu_0(\alpha)$.
	\label{lemma:consistency_est_ps}
\end{lemma}
\renewcommand*{\proofname}{\textbf{Proof of Lemma \ref{lemma:consistency_est_ps}}}
\begin{proof}
	Consider the augmented estimated equations defined as $\Psi_n(\vect\theta) = \sum_{i = 1}^N \vect \psi(\mb Y_i, \mb A_i, \mb L_i;\vect\theta)$, where
	\[
	\vect \psi(\mb Y_i, \mb A_i, \mb L_i;\vect\theta) =
	\begin{pmatrix}
	\vect \psi_\gamma(\mb L_i, \mb A_i; \vect\gamma) \\
	\psi_{0, \alpha}(\mb Y_i, \mb A_i, \mb L_i; \mu^0, \vect\gamma) \\
	\psi_{1, \alpha}(\mb Y_i, \mb A_i, \mb L_i; \mu^1, \vect\gamma) 
	\end{pmatrix}_{(p + 2)\times1}
	\]
	where $\vect\theta = (\vect \gamma^T, \mu^0, \mu^1)^T$ and $p$ is the number of parameters of the parametric propensity score model. Note that $\psi_{a, \alpha}$ is now a function of $\vect\gamma$ since it uses the estimated propensity score. Denote the vector that solves $\Psi_n(\vect\theta) = 0$ as $\widehat{\vect\theta}$. The first $p$ elements of $\widehat{\vect\theta}$ correspond to estimators of the propensity score model parameters, which are consistent for $\vect\gamma_0$.
	Since $\widehat Y_N^L(\alpha)$ based on the true propensity score is consistent, $f(\mb a_i|\mb l_i;\vect\gamma)$ is differentiable in $\vect\gamma$ and therefore continuous, and $\vect\gamma \overset{p}{\rightarrow} \vect\gamma_0$, the last two elements of $\widehat{\vect\theta}$ which correspond to $\widehat{Y}^L(0, \alpha), \widehat{Y}^L(1, \alpha)$ using the estimated propensity score are consistent estimators of $\overline{Y}^L(0, \alpha), \overline{Y}^L(0, \alpha)$.
\end{proof}
%
\renewcommand*{\proofname}{\textbf{Proof of Theorem \ref{theorem:asymptotic_normality_estimated_ps}}}
\begin{proof}
	\noindent
	We will again use Theorem 5.41 of \cite{Vaart1998}.
	Since consistency has been established in Lemma \ref{lemma:consistency_est_ps}, showing the four conditions stated in the proof of Theorem \ref{theorem:asymptotic_normality} for the augmented $\vect\psi$ will establish asymptotic normality. Denote $\vect \theta_0 = (\vect \gamma_0^T, \vect\mu_0(\alpha)^T)^T$.
	
	\noindent
	\textit{Proof of (i).}
	By the conditions of the theorem, $\vect \gamma \rightarrow \vect\psi_\gamma(\mb l_i, \mb a_i;\vect\gamma)$ is twice continuously differentiable. This implies that
	$\psi_{a, \alpha}(\mb y_i, \mb l_i, \mb a_i; \mu_0(a, \alpha), \vect\gamma)$, $a = 0, 1$ are three times continuously differentiable with respect to $\vect \gamma$. Therefore, the second order partial derivatives with respect to $\vect\gamma$ exist and are continuous.
	Moreover, since $\vect\psi_\gamma(\mb l_i, \mb a_i;\vect\gamma)$ is not a function of $\mu^a$, and using (\ref{eq:proof_psi_dot}), the second partial derivatives with respect to elements of $\vect\mu = (\mu^0, \mu^1)$ exist and are continuous.
	Lastly, all second order derivatives with respect to an element of $\vect\mu$ and an element of $\vect\gamma$ exist and are 0, and therefore continuous.
	This shows that $\vect\theta \rightarrow \vect\psi(\mb y_i, \mb l_i, \mb a_i;\vect\theta)$ is twice continuously differentiable.
	
	\noindent
	\textit{Proof of (ii).} We want to show that $E_{\mb{Y}_i, \mb L_i, \mb A_i}\| \psi(\mb{Y}_i, \mb L_i, \mb A_i; \vect \theta_0) \|_2^2 < \infty$. But
	\begin{align*}
	E_{\mb{Y}_i, \mb L_i, \mb A_i} &\| \psi(\mb{Y}_i, \mb L_i, \mb A_i; \vect \theta_0)\|_2^2 = \\
	& E_{\mb L_i, \mb A_i} \| \vect\psi_\gamma(\mb L_i, \mb A_i; \vect\gamma_0) \|^2_2 + \sum_{a \in \{0, 1\}} E_{\mb{Y}_i, \mb L_i, \mb A_i} \| \psi_{a, \alpha}(\mb Y_i, \mb A_i, \mb L_i;\mu_0(a, \alpha)) \|^2,
	\end{align*}
	where the first term is finite from the assumptions on the propensity score model, and the terms in the summation are finite from (\ref{eq:proof_bound_c}).
	
	\noindent
	\textit{Proof of (iii).}
	We want to show that the matrix $E_{\mb Y_i, \mb L_i, \mb A_i}\left[\overset{\cdot}{\psi}(\mb y_i, \mb l_i, \mb a_i; \vect \theta_0)\right]$ exists and is non singular. We have
	\[
	\overset{\cdot}{\psi}(\mb y_i, \mb l_i, \mb a_i; \vect \theta) =
	\begin{pmatrix}
	\frac{\partial}{\partial \vect\gamma^T}\vect\psi_\gamma(\mb L_i, \mb A_i; \vect\gamma)_{p \times p} & 0_{p\times 1} & 0_{p \times 1} \\
	\frac{\partial}{\partial \vect\gamma^T}\vect\psi_{0, \alpha}(\mb Y_i, \mb L_i, \mb A_i; \mu^0, \vect\gamma)_{1 \times p} & - 1 & 0 \\
	\frac{\partial}{\partial \vect\gamma^T}\vect\psi_{1, \alpha}(\mb Y_i, \mb L_i, \mb A_i; \mu^1, \vect\gamma)_{1 \times p} & 0 & - 1
	\end{pmatrix},
	\]
	where the the 0's in the top row are because $\vect\psi_\gamma$ is not a function of $\mu^0, \mu^1$. We have assumed that $E\left[\frac{\partial}{\partial \vect\gamma^T}\vect\psi_\gamma(\mb L_i, \mb A_i; \vect\gamma_0)\right]$ exists and we will show that $$E\left[\frac\partial{\partial \gamma^T}\psi_{a, \alpha}(\mb Y_i, \mb L_i, \mb A_i;\mu_0(a, \alpha), \vect\gamma_0)\right]$$ exists for $a = 0, 1$. \\
	
	\noindent
	\textit{Showing that $E\left[\frac\partial{\partial \gamma^T}\psi_{a, \alpha}(\mb Y_i, \mb L_i, \mb A_i;\mu_0(a, \alpha))\right] < \infty$ for $a = 0, 1$}.
	
	\noindent Note that even if the estimates of $\vect\gamma$ were used to define the counterfactual treatment allocation $\numerator$, it is considered fixed as a function of $\vect\gamma$, since it is used to represent a fixed realistic treatment allocation program.
	\begin{align*}
	\frac{\partial}{\partial\gamma_k} &
	\psi_{a, \alpha}(\mb O_i;\mu^a, \vect\gamma) = 
	\left(\frac1{n_i}\sum_{j = 1}^{n_i} \numerator I(A_{ij} = a) Y_{ij} \right)
	\left(\frac{\partial}{\partial\gamma_k} \frac1\denominator \right) \\
	=& - \left(\frac1{n_i}\sum_{j = 1}^{n_i} \numerator I(A_{ij} = a) Y_{ij} \right)
	\left( \frac{\frac{\partial}{\partial\gamma_k} \log \denominator}{\denominator} \right) \\
	= & - \psi_\gamma^k\left(\tilde{\mb O}_i; \vect\gamma \right) \left(\frac1{n_i}\sum_{j = 1}^{n_i} \frac\numerator\denominator I(A_{ij} = a) Y_{ij} \right),
	\numberthis \label{eq:psi_a_wrt_gamma}
	\end{align*}
	where $\psi_\gamma^k\left(\tilde{\mb O}_i; \vect\gamma \right)$ is the $k^{th}$ component of $\vect\psi_\gamma\left(\tilde{\mb O}_i; \vect\gamma \right)$ for which $E_{F_0}\left[ \psi_\gamma^k\left(\tilde{\mb O}_i; \vect\gamma_0 \right) \right] < \infty$ (Lemma \ref{lemma:exp_score}). Also,
	$
	\left|\frac\numerator\denominator I(A_{ij} = a) Y_{ij} \right| < M/\delta_o
	$
	using the conditions of Theorem \ref{theorem:asymptotic_normality}. So, we have shown that $E\left[\frac\partial{\partial \gamma^T}\psi_{a, \alpha}(\mb Y_i, \mb L_i, \mb A_i;\mu_0(a, \alpha))\right] < \infty$.
	
	From this, we conclude that $E_{F_0}\left[ \overset{\cdot}{\psi}(\mb y_i, \mb l_i, \mb a_i; \vect \theta) \right]$ exists. 
	Furthermore, from the theorem assumptions we have that $E\left[\frac{\partial}{\partial \vect\gamma^T}\vect\psi_\gamma(\mb L_i, \mb A_i; \vect\gamma_0)\right]$ is non-singular and the rows of $\partial\vect\psi_\gamma / \partial \gamma^T$ are linearly independent. The bottom two rows are linearly independent to the rest since they are the only ones to include non-zero elements in the last two columns. From this, we conclude that the rows of $E\left[\overset{\cdot}{\psi}(\mb Y_i, \mb L_i, \mb A_i; \vect \theta_0)\right]$ are linearly independent, and the matrix is full rank and non-singular.

	\noindent
	\textit{Proof of (iv).}
	We need to show that $\exists$ integrable function $\alpha(\mb o_i)$ fixed, such that $\alpha(\mb o_i)$ dominates all the second order partial derivatives of $\vect\psi(\mb o_i;\vect\theta)$. 
	Therefore, we need to show that for $k, l \in \{1, 2, \dots, p\}$, $a \in \{0, 1\}$:
	\begin{enumerate}
		\item \( \left| \displaystyle \frac{\partial^2
			\vect\psi_\gamma\left(\tilde{\mb o}_i ;\vect\gamma\right)}
		{\partial \gamma_k \partial \gamma_l} \right|
		\leq \alpha_{kl}(\mb o_i ) \), \label{wtp1}
		\item \( \left| \displaystyle \frac{\partial^2 \vect\psi_\gamma\left(\tilde{\mb o}_i ;\vect\gamma\right)} {\partial \gamma_k \partial \mu^a} \right| \leq \alpha_k^a(\mb o_i) \), \label{wtp2}
		\item \( \left| \displaystyle \frac{\partial^2 \vect\psi_\gamma\left(\tilde{\mb o}_i ;\vect\gamma\right)} {\partial \mu^{a_1} \partial \mu^{a_2}} \right| \leq \alpha^{a_1a_2}(\mb o_i) \), \label{wtp3}
		\item \( \left| \displaystyle \frac{\partial^2 \psi_{a, \alpha}(\mb o_i;\mu^a, \vect\gamma)}{\partial \mu^{a_1} \partial \mu^{a_2}} \right| \leq \xi^{a_1a_2}(\mb o_i) \), \label{wtp4}
		\item \( \left| \displaystyle \frac{\partial^2 \psi_{a, \alpha}(\mb o_i;\mu^a, \vect\gamma)}{\partial \mu^{a_1} \partial \gamma_k} \right| \leq \xi^{a_1}_k(\mb o_i) \), \label{wtp5}
		\item \( \left| \displaystyle \frac{\partial^2 \psi_{a, \alpha}(\mb o_i;\mu^a, \vect\gamma)}{\partial \gamma_k \partial \gamma_l} \right| \leq \xi_{kl}(\mb o_i) \), \label{wtp6}
	\end{enumerate}
	for $\vect\theta$ in a neighborhood of $\vect\theta_0$,
	where
	$\alpha_{kl}({\mb o}_i), \alpha_k^a({\mb o}_i),
	\alpha^{a_1a_2}({\mb o}_i), \xi^{a_1a_2}({\mb o}_i),
	\xi^a_k({\mb o}_i), \xi_{kl}({\mb o}_i)$
	are $F_0$-integrable. If we show the above, by setting $\alpha(\mb o_i) = \max_{k,l,a}\{\alpha_{kl}({\mb o}_i),
	\alpha_k^a({\mb o}_i), \alpha^{a_1a_2}({\mb o}_i),$ $\xi^{a_1a_2}({\mb o}_i), \xi^a_k({\mb o}_i),
	\xi_{kl}({\mb o}_i) \}$ we have that all second order partial derivatives are dominated by the $F_0$ integrable $\alpha(\mb o_i)$.
	
	Since $\vect\psi_\gamma\left(\tilde{\mb o}_i;\vect\gamma\right)$ is not a function of $\mu^a$, conditions \ref{wtp2}, \ref{wtp3} are easy to satisfy by setting $\alpha_k^a({\mb o}_i) = \alpha^{a_1a_2}({\mb o}_i) = 0$.
	The same is true for conditions \ref{wtp4}, \ref{wtp5}, since $\partial \psi_{a, \alpha}(\mb o_i;\mu_a,\vect\gamma)/\partial\mu^{a_1} = - I(a = a_1)$ and therefore all second order derivatives that include at least one derivative with respect to $\mu^{a_1}$ will be equal to 0. So we can set $\xi^{a_1a_2}({\mb o}_i) = \xi^a_k({\mb o}_i) = 0$.
	
	From the assumptions of the theorem, we know that $\exists$ $\overset{\cdot \cdot}{\psi}_\gamma(\mb l_i, \mb a_i)$ integrable such that
	\( \displaystyle \left| \frac{\partial^2 \vect\psi_\gamma(\mb l_i, \mb a_i;\gamma)}{\partial \gamma_k \partial \gamma_l} \right| \leq \overset{\cdot \cdot}{\psi}_\gamma(\mb l_i, \mb a_i) \), for all $\vect \gamma$ in a neighborhood of $\vect\gamma_0$. Then, $\alpha_{kl}({\mb o}_i) = \overset{\cdot \cdot}{\psi}_\gamma(\mb o_i)$ satisfy condition \ref{wtp1}. 
	Since $\gamma_0$ is in an open subset of the Euclidean space, there exists $\epsilon > 0$ such that the second partial derivatives of $\vect\psi_\gamma$ are dominated by $\overset{\cdot\cdot}{\psi}$ for all $\vect\gamma \in \mathcal{N}^\epsilon(\vect\gamma_0) = \{\vect\gamma: \|\vect\gamma - \vect\gamma_0\| < \epsilon\}$, subset of the parameter space. Let $\overline{\mathcal N}^{\epsilon / 2}(\vect\gamma_0) = \{ \vect\gamma: \|\vect\gamma - \vect\gamma_0\| \leq \epsilon/2 \} \subset \mathcal{N}^\epsilon(\vect\gamma_0)$. Then, $\overline{\mathcal N}^{\epsilon / 2}(\vect\gamma_0)$ is a compact subset of the Euclidean space.
	
	We will show that for $\vect\gamma \in \overline{\mathcal N}^{\epsilon / 2}(\vect\gamma_0)$
	the second order partial derivatives in \ref{wtp6} are bounded by an integrable function. First, let's acquire their form:
	\small{
	\begin{align*}
	& \frac{\partial^2 \psi_{a, \alpha}(\mb O_i;\mu^a,\vect\gamma)}{\partial\gamma_k\partial\gamma_l} \\
	= &
	\frac{\partial}{\partial\gamma_k} \left[
	- \psi_\gamma^l\left(\tilde{\mb O}_i; \vect\gamma \right) \left(\frac1{n_i}\sum_{j = 1}^{n_i} \frac\numerator\denominator I(A_{ij} = a) Y_{ij} \right)\right] \\
	= & - \frac\partial{\partial\gamma_k} \psi_\gamma^l\left(\tilde{\mb O}_i; \vect\gamma \right) \left(\frac1{n_i}\sum_{j = 1}^{n_i} \frac\numerator\denominator I(A_{ij} = a) Y_{ij} \right) \\
	& - \psi_\gamma^l\left(\tilde{\mb O}_i; \vect\gamma \right)
	\frac{\partial}{\partial\gamma_k} \psi_{a, \alpha}(\mb O_i;\mu^a, \vect\gamma) \\
	= & - \frac\partial{\partial\gamma_k} \psi_\gamma^l\left(\tilde{\mb O}_i; \vect\gamma \right) \left(\frac1{n_i}\sum_{j = 1}^{n_i} \frac\numerator\denominator I(A_{ij} = a) Y_{ij} \right) \\
	& + \psi_\gamma^l\left(\tilde{\mb O}_i; \vect\gamma \right)
	\psi_\gamma^k\left(\tilde{\mb O}_i; \vect\gamma \right) \left(\frac1{n_i}\sum_{j = 1}^{n_i} \frac\numerator\denominator I(A_{ij} = a) Y_{ij} \right) \\
	= & \left[ \psi_\gamma^l\left(\tilde{\mb O}_i; \vect\gamma \right)
	\psi_\gamma^k\left(\tilde{\mb O}_i; \vect\gamma \right) - \frac\partial{\partial\gamma_k} \psi_\gamma^l\left(\tilde{\mb O}_i; \vect\gamma \right)  \right]
	\left[\frac1{n_i}\sum_{j = 1}^{n_i} \frac\numerator\denominator I(A_{ij} = a) Y_{ij} \right]
	\end{align*}}
	where the first and third equation use (\ref{eq:psi_a_wrt_gamma}), and the second equation is an application of the chain rule. Then
	\begin{align}
	\left| \frac{\partial^2 \psi_{a, \alpha}(\mb O_i;\mu^a,\vect\gamma)}{\partial\gamma_k\partial\gamma_l} \right| <
	\frac M{\delta_o} \left| \psi_\gamma^l\left(\tilde{\mb O}_i; \vect\gamma \right)
	\psi_\gamma^k\left(\tilde{\mb O}_i; \vect\gamma \right) - \frac\partial{\partial\gamma_k} \psi_\gamma^l\left(\tilde{\mb O}_i; \vect\gamma \right) \right|
	\label{eq:psi_a_gamma_gamma}
	\end{align}
	For all $k, l\in \{1, 2, \dots, p\}$, 
	$\psi_\gamma^l\left(\tilde{\mb O}_i; \vect\gamma \right)$, $\frac\partial{\partial\gamma_k} \psi_\gamma^l\left(\tilde{\mb O}_i; \vect\gamma \right)$ are differentiable and therefore continuous in $\vect\gamma$, implying that the function on the right-hand side of (\ref{eq:psi_a_gamma_gamma}) is continuous in $\vect\gamma$.
	
	Define
	\(\displaystyle
	g(\vect\gamma) =
	E_{F_0} \left| \psi_\gamma^l\left(\tilde{\mb O}_i; \vect\gamma \right)
	\psi_\gamma^k\left(\tilde{\mb O}_i; \vect\gamma \right) - \frac\partial{\partial\gamma_k} \psi_\gamma^l\left(\tilde{\mb O}_i; \vect\gamma \right) \right|
	\).
	Then $g(\vect\gamma)$ is continuous in $\vect\gamma$. But since $\overline{\mathcal{N}}^{\epsilon/2}(\vect\gamma_0)$ is a compact set, $g(\vect\gamma)$ is bounded in $\overline{\mathcal{N}}^{\epsilon/2}(\vect\gamma_0)$, and in fact achieves a maximum. Let 
	$$
	\xi_{kl}(\mb o_i) = \xi_{kl} = \frac M{\delta_o}
	\max\left\{g(\vect\gamma), \vect\gamma \in \overline{\mathcal{N}}^{\epsilon/2}(\vect\gamma_0) \right\}.
	$$
	Then \(\displaystyle \left| \frac{\partial^2 \psi_{a, \alpha}(\mb O_i;\mu^a,\vect\gamma)}{\partial\gamma_k\partial\gamma_l} \right| < \xi_{kl}(\mb o_i), \ \forall \vect\gamma \in \overline{\mathcal L}^{\epsilon / 2}(\vect\gamma_0) \), and $\xi_{kl}$ is integrable since it is a constant function.
	
	Then, set $\alpha(\mb o_i) = \max\{\alpha_{kl}({\mb o}_i),
	\alpha_k^a({\mb o}_i), \alpha^{a_1a_2}({\mb o}_i), \xi^{a_1a_2}({\mb o}_i),
	\xi^a_k({\mb o}_i), \xi_{kl}({\mb o}_i) \}$, and
	all second order partial derivatives are dominated by the $F_0$-integrable $\alpha(\mb o_i)$, for all $\vect\theta \in \mathcal{N}^{\epsilon / 2}(\vect\theta_0) = \left\{\vect\theta: \|\theta - \theta_0\| < \epsilon/2 \right\}$.
	
	From Theorem 5.41 of \cite{Vaart1998}, we have that $$\sqrt{n}\left(\widehat\theta - \theta_0 \right) \overset{N\rightarrow\infty}{\rightarrow} N(0, Q(\theta_0)),$$ where 
	$$
	Q(\theta_0) = A(\theta_0)^{-1}B(\theta)[A(\theta_0)^{-1}]^T$$
	for
	$A(\theta_0) = E\left[\overset{\cdot}{\psi}(\mb O_i; \vect \theta_0)\right]$, and
	$B(\theta_0) = E\left[{\psi}(\mb O_i; \vect \theta_0){\psi}(\mb O_i; \vect \theta_0)^T \right]$.
	
	However, we are only interested in the bottom-right $2\times2$ submatrix of $Q(\theta_0)$ which corresponds to the asymptotic variance of $(\widehat\mu_0, \widehat\mu_1)^T$ when the propensity score model is estimated. Note that $A(\theta)$, $B(\theta)$ can be rewritten as
	\begin{align*}
	A(\theta) = \begin{bmatrix} A_{11} & 0 \\ A_{21} & - I_2 \end{bmatrix} & \text{ where }
	A_{11} = E\left[\frac{\partial \psi_\gamma}{\partial \gamma^T} \right]_{p\times p}
	A_{21} = E\left[\frac{\partial (\psi_0, \psi_1)^T}{\partial \gamma^T} \right]_{2\times p}, \\ 
	B(\theta) = \begin{bmatrix} B_{11} & B_{12} \\ B_{21} & B_{22} \end{bmatrix} & \text{ for }
	B_{11} = E\left[ \psi_\gamma \psi_\gamma^T\right]_{p\times p}
	B_{12} = E[\psi_\gamma \psi_0, \psi_\gamma\psi_1]_{p\times 2},\ 
	B_{21} = B_{12}^T, \text{ and } \\
	B_{22} = & E\begin{bmatrix} \psi_0^2 & \psi_0\psi_1 \\ \psi_1\psi_0 & \psi_1^2 \end{bmatrix},
	\end{align*}
	where the arguments $(\mb Y_i, \mb L_i, \mb A_i)$ have been suppressed. Then,
	\begin{align*}
	& A(\theta)^{-1} B(\theta) [A(\theta)^{-1}]^T = 
	\begin{bmatrix} A_{11}^{-1} & \mb 0 \\ A_{21}A_{11}^{-1} & -I_2 \end{bmatrix}
	\begin{bmatrix} B_{11} & B_{12} \\ B_{21} & B_{22} \end{bmatrix}
	\begin{bmatrix} A_{11}^{-1} & \mb 0 \\ A_{21}A_{11}^{-1} & -I_2 \end{bmatrix}^T \\
	=& \begin{bmatrix} A_{11}^{-1}B_{11} & A_{11}^{-1}B_{12} \\ 
	A_{21}A_{11}^{-1}B_{11} - B_{21} & A_{21}A_{11}^{-1}B_{12} - B_{22} \end{bmatrix} \begin{bmatrix} (A_{11}^{-1})^T & (A_{21}A_{11}^{-1})^T \\ \mb 0 & - I_2 \end{bmatrix} \\
	=& \begin{bmatrix} - I_p & -B_{11}^{-1}B_{12} \\ 
	- A_{21} - B_{21} & - A_{21}B_{11}^{-1}B_{12} - B_{22} \end{bmatrix}
	\begin{bmatrix} -(B_{11}^{-1})^T & - (A_{21}B_{11}^{-1})^T \\ \mb 0 & - I_2 \end{bmatrix} \tag{Since $A_{11} = - B_{11}$.} \\
	=& \begin{bmatrix} \dots & \dots \\ \dots & (A_{21} + B_{21})B_{11}^{-1}A_{21}^T +  A_{21}B_{11}^{-1} B_{12} + B_{22} \end{bmatrix}
	\tag{$B_{11}$ symmetric $\Rightarrow B_{11}^{-1}$ symmetric} \\ 
	=& \begin{bmatrix} \dots & \dots \\ \dots & A_{21}B_{11}^{-1}A_{21}^T + A_{21}B_{11}^{-1}B_{12} + \left(A_{21}B_{11}^{-1}B_{12}\right)^T + B_{22}
	\end{bmatrix} \tag{$B_{21} = B_{12}^T$}
	\end{align*}
	So the asymptotic covariance matrix of $(\widehat\mu_0, \widehat\mu_1)$ is equal to $$A_{21}B_{11}^{-1}A_{21}^T + A_{21}B_{11}^{-1}B_{12} + \left(A_{21}B_{11}^{-1}B_{12}\right)^T + B_{22}.$$
\end{proof}

\section{Asymptotic variance of the population average potential outcome estimator}
\label{supp_sec:Vmu}
Denote $[ V(\vect\mu_0(\alpha)) ]_{ij}$ the $ij$ element of the covariance matrix, and remember that $\mu_0(a, \alpha) = E_{F_0}[\overline{Y}_i(a, \alpha)]$. Then
\begin{align*}
[V(\vect\mu_0(\alpha))]_{(a+1)(a+1)}
= & E_{F_0} \left[ \left( \frac1{n_i} \sum_{j = 1}^{n_i}
\frac{\counteralloc(\mb A_{i, -j} | A_{ij} = a, \mb L_i)}
{f_{\mb A| \mb L, i}(\mb A_i | \mb L_i)} I(A_{ij} = a) Y_{ij} - \mu_0(a, \alpha) \right)^2 \right] \\
= & E_{F_0} \left[ \left( \frac1{n_i} \sum_{j = 1}^{n_i}
\frac{\counteralloc(\mb A_{i, -j} | A_{ij} = a, \mb L_i)}
{f_{\mb A| \mb L, i}(\mb A_i | \mb L_i)} I(A_{ij} = a) Y_{ij} \right)^2  \right] +
\mu_0(a, \alpha)^2 - \\
& - 2\mu_0(a, \alpha) E_{F_0}\left[ \frac1{n_i} \sum_{j = 1}^{n_i}
\frac{\counteralloc(\mb A_{i, -j} | A_{ij} = a, \mb L_i)}
{f_{\mb A| \mb L, i}(\mb A_i | \mb L_i)} I(A_{ij} = a) Y_{ij}  \right] \\
=& E_{F_0} \big[\overline{Y}_i^L(a, \alpha)^2 \big] + \mu_0(a, \alpha)^2 -
2\mu_0(a, \alpha) E_{F_0}\big[\overline{Y}^L_i(a, \alpha)\big] \\
=&
Var_{F_0}\big[\overline{Y}_i^L(a, \alpha)\big]
\end{align*}
\begin{align*}
[V(\vect\mu_0(\alpha))]_{12} 
=& E_{F_0} \left[
\left( \widehat{Y}_i^L(0, \alpha) - \mu_0(0, \alpha) \right)
\left( \widehat{Y}_i^L(1, \alpha) - \mu_0(1, \alpha) \right) \right] \\
=& E_{F_0}\left[ \widehat{Y}_i^L(0, \alpha) \widehat{Y}_i^L(1, \alpha)\right] - \mu_0(0, \alpha) E_{F_0}\left[\widehat{Y}_i^L(1, \alpha) \right] - \\
& - \mu_0(1, \alpha) E_{F_0}\left[ \widehat{Y}_i^L(0, \alpha)\right] + \mu_0(0, \alpha) \mu_0(1, \alpha) \\
=& E_{F_0}\left[ \widehat{Y}_i^L(0, \alpha) \widehat{Y}_i^L(1, \alpha)\right] - E_{F_0}\left[ \widehat{Y}_i^L(0, \alpha) \right]  E_{F_0}\left[\widehat{Y}_i^L(1, \alpha)\right] \\
=& Cov_{F_0}\left(\overline{Y}_i^L(0, \alpha), \overline{Y}_i^L(1, \alpha)\right)
\end{align*}

\section{Population average potential outcome definitions in the literature}
\label{supp_sec:population_estimand}

Assuming partial interference, \cite{Hudgens2008}, and \cite{Tchetgen2012} defined the population average potential outcome as an average of the group-level potential outcomes
$\overline{Y}(a;\alpha) = \frac 1N \sum_{i = 1}^N \overline{Y}_i(a; \alpha)$.
On the other hand, \cite{Liu2016} define the population average potential outcomes without assuming partial interference (and therefore without assuming the existence of interference clusters) as the average of the individual average potential outcomes. However, their asymptotic results are based on the assumption of partial interference, under which the population average potential outcome can be written as 
\[ \overline{Y}^{Liu}(a;\alpha) = \frac 1{\sum_{i = 1}^N n_i} \sum_{i = 1}^N\sum_{j = 1}^{n_i} \overline{Y}_{ij}(a; \alpha) =  \sum_{i = 1}^N \frac {n_i}{\sum_{i = 1}^N n_i} \overline{Y}_i(a; \alpha).\]
Therefore, the estimand of \cite{Liu2016}, if partial interference is assumed, is equal to a weighted average of the group average potential outcomes with weights proportional to the number of individuals in the cluster.

Estimators for both quantities can be written in the form 
\begin{equation}
\widehat{Y}(a; \alpha) = \sum_{i = 1}^N \frac{d_i}{\sum_{i = 1}^N d_i} \widehat{Y}_i(a;\alpha),\ d_i > 0,
\label{eq:di_estimator}
\end{equation}
where $\widehat{Y}_i(a; \alpha)$ is an unbiased estimator of the group average potential outcome for cluster $i$. The difference of the population average estimators lies in the specification of $d_i$, where $d_i = 1$ and $d_i = n_i$ accordingly, for the two definitions of population average potential outcome.

\begin{proposition}
	Under the assumption of partial interference (which is also assumed by \cite{Liu2016} in their asymptotic results), all population average potential outcome estimators of the form (\ref{eq:di_estimator}) for which $d_i > 0$ does not depend on $N$, $E_{F_0}[d_i] < \infty$, and $d_i \amalg \overline{Y}_i(a; \alpha)$ are consistent for $E_{F_0} \left[ \overline Y_i(a;\alpha) \right]$.
	\label{prop:di_estimators}
\end{proposition}

\renewcommand*{\proofname}{\textbf{Proof of Proposition \ref{prop:di_estimators}}}

\begin{proof}
	This can be shown by considering the estimating equation
	$$\sum_{i = 1}^N G_i(\mb Y_i, \mb L_i, \mb A_i, d_i; \mu) = 0,$$
	where
	$$
	G_i(\mb Y_i, \mb L_i, \mb A_i, d_i; \mu) = d_i \left[ \widehat{Y}_i(a; \alpha) - \mu \right],
	$$
	The solution to this equation is
	\[ \displaystyle \widehat{\mu} = \sum_{i = 1}^N \frac{d_i}{\sum_{i = 1}^N d_i} \widehat{Y}_i(a, \alpha), \]
	and the solution to $\int G_i(\mb y_i, \mb l_i, \mb a_i, d_i;\mu) \mathrm{d}F_0(\mb y_i, \mb l_i, \mb a_i) = 0$ is
	\begin{align*}
	& \frac{E[d_i \widehat{Y}_i(a; \alpha)]}{E[d_i]} =
	\frac{E\left\{ E \big[d_i \widehat{Y}_i(a; \alpha) | d_i\big] \right\}}{E[d_i]} = \frac{E\left\{ d_i E\big[\widehat{Y}_i(a;\alpha)\big] \right\}}{E[d_i]} \\
	& = \frac{E\big[d_i \overline{Y}_i(a;\alpha)\big]}{E[d_i]}
	= E_{F_0} \left[ \overline{Y}_i(a;\alpha) \right]
	= \mu_0(a, \alpha),
	\end{align*}
	since $d_i \amalg \overline{Y}_i(a;\alpha)$.
	
	Since $G_i$ is monotone in $\mu$, both $\sum_{i = 1}^N G_i$ and $\int G_i$ are monotone in $\mu$ which implies uniqueness of the roots and establishes $\widehat\mu \overset{p}{\rightarrow} \mu_0(a, \alpha)$.
\end{proof}

Based on this, assuming $n_i \amalg \overline{Y}_i(a;\alpha)$ both estimators are consistent for the same quantity. However, when the propensity score is known, the weighting scheme $d_i = c$, constant, leads to the asymptotically most efficient estimator among all of the estimators of the form (\ref{eq:di_estimator}), based on the following proposition. Since two estimators using $d_i$ and $d_i' = cd_i$ are exactly the same, the estimator (\ref{eq:di_estimator}) for $d_i = 1$ is the asymptotically efficient estimator.

\begin{proposition}
	Assuming that the conditions of Theorem \ref{theorem:asymptotic_normality} and Proposition \ref{prop:di_estimators} hold, and $\exists M_d$ such that $d_i < M_d, \forall i$, then $\widehat{Y}(a;\alpha)  = \frac1N \sum_{i = 1}^N \widehat{Y}_i(a;\alpha)$ is the asymptotically most efficient estimator of $\mu_0(a, \alpha)$ among all estimators of the class (\ref{eq:di_estimator}).
	\label{prop:di_variances}
\end{proposition}

\renewcommand*{\proofname}{\textbf{Proof of Proposition \ref{prop:di_variances}}}
\begin{proof}
	Based on Proposition \ref{prop:di_estimators}, $\widehat{\mu}_{\vect d}(a;\alpha) = \sum_{i = 1}^N \frac{d_i}{\sum d_i} \widehat{Y}_i(a;\alpha)$ are consistent for $\mu_0= \mu_0(a, \alpha)$.
	Since $G_i(\mb Y_i, \mb L_i, \mb A_i, d_i; \mu) = d_i \left[ \widehat{Y}_i(a; \alpha) - \mu \right]$ is monotone decreasing in $\mu$ with $\frac{\partial}{\partial \mu} G_i(\mb Y_i, \mb L_i, \mb A_i, d_i; \mu) = - d_i < 0$, we have that $\mu_{\vect d}$ and $\mu_0$ are isolated roots of $\sum_{i = 1}^N G_i = 0$ and $\int G_i = 0$. Also, $E[\frac{\partial}{\partial \mu} G_i(\mb Y_i, \mb L_i, \mb A_i, d_i; \mu)] = -E[d_i] \neq 0$. Lastly, from (\ref{eq:proof_bound_c}) and $d_i < M_d$ we have that $\int G_i^2$ is bounded by $M_d^2c^2$. We can straightforwardly use M-estimation theory to acquire the asymptotic variance.
	Lemma A in section 7.2.1 of \cite{Serfling1980},
	$\sqrt{n}(\widehat{\mu}_{\vect d} - \mu_0) \overset{d}{\rightarrow} N(0, \sigma^2(\vect d))$, where
	$\sigma^2(\vect d) = E_{F_0}[G_i^2(\cdot, d_i;\mu_0)] / E_{F_0}^2[d_i]$.
	
	We will show that $\sigma^2(\vect d)$ is minimized when $d_i = 1, \forall i$. Since $d_i \amalg \overline{Y}_i(a;\alpha)$, we have that $d_i^2 \amalg (\overline{Y}_i(a;\alpha) - \mu_0)^2$. Then \( \displaystyle
	\sigma^2_{\vect d}
	=  \frac{E_{F_0}\left[ d_i^2 ( \widehat{Y}(a; \alpha) - \mu_0)^2 \right]}{[E_{F_0}d_i]^2}
	= \frac{E_{F_0}\left[d_i^2 \right]}{ E_{F_0}^2[d_i]}
	E_{F_0}\left[ (\widehat{Y}_i(a;\alpha) - \mu_0)^2 \right]
	= \frac{E_{F_0}\left[d_i^2 \right]}{ E_{F_0}^2[d_i]} \sigma^2_{\vect 1}
	\),
	where $\sigma^2_{\vect 1}$ is the asymptotic variance of the estimator for $d_i = 1$. From Jensen's inequality, and since $\phi(x) = x^2$ is a convex function, we have that $E^2_{F_0}[d_i] \leq E_{F_0}[d_i^2]$, which establishes $\sigma^2_{\vect d} \geq \sigma^2_{\vect 1}$. Equality holds if and only if all values $d_i$ are equal.
\end{proof}

\section{Calculating cluster-intercept for a specific cluster average propensity of treatment}
\label{supp_sec:computation}

As described in section \ref{subsec:sim_true_po}, $\xi_i^\alpha$ is chosen such that
\begin{equation}
\frac 1{n_i}\sum_{j = 1}^{n_i} \counteralloc(A_{ij} = 1 | \mb L_i) = \alpha,
\label{app_eq:choose_xi}
\end{equation}
where
\( \displaystyle
\mathrm{logit}\counteralloc(A_{ij} = 1 | \mb L_i) = \xi_i^\alpha + L_{ij} \vect\delta,
\)
and $L_{ij} = (L_{1ij}, L_{2ij}, \dots, L_{pij})^T$ the value of the $p$ predictors of the propensity score model. Then, (\ref{app_eq:choose_xi}) can be rewritten as
\begin{align*}
& \frac 1{n_i}\sum_{j = 1}^{n_i} \counteralloc(A_{ij} = 1 | \mb L_i) =
\frac1{n_i} \sum_{j = 1}^{n_i}
\frac{ \exp\left\{ \xi_i^\alpha + L_{ij} \vect\delta \right\}}
{1 + \exp\left\{ \xi_i^\alpha + L_{ij} \vect\delta \right\}}
=  \frac1{n_i} \sum_{j = 1}^{n_i}
\frac{ \exp\left\{L_{ij} \vect\delta \right\}}
{\exp\left\{ \xi_i^\alpha \right\} +
	\exp\left\{L_{ij} \vect\delta \right\}}
= \alpha \\
& \iff \left| \frac1{n_i} \sum_{j = 1}^{n_i}
\frac{ \exp\{L_{ij} \vect\delta \}}
{\exp\{ - \xi_i^\alpha \} + \exp \{L_{ij} \vect\delta \}}
- \alpha \right| = 0
\end{align*}
Since the only unknown is $\xi_i^\alpha$, we use optimization techniques and set $\xi_i^\alpha$ to be the value $\xi$ at which the function 
$$
g(\xi) = \left| \frac1{n_i} \sum_{j = 1}^{n_i}
\frac{ \exp\left\{L_{ij} \vect\delta \right\}}
{\exp\{ -\xi \} + \exp\{L_{ij} \vect\delta \}}
- \alpha \right|
$$
is minimized.


\bibliographystyle{biom}
\bibliography{Interference}

\end{document}